\documentclass[letterpaper]{article}

\usepackage{amsmath}
\usepackage{amssymb}
\usepackage{amsthm}
\usepackage{thmtools}
\usepackage{natbib}
\usepackage{hyperref}
\usepackage{cleveref}
\usepackage{xcolor}
\usepackage{geometry}
\usepackage{tablefootnote}
\usepackage{booktabs}

\newcommand{\edit}[1]{\textcolor{red}{[#1]}}
\newcommand\set[1]{{ \{\, #1 \,\} }}
\newcommand\ceil[1]{\left\lceil#1\right\rceil}
\newcommand\floor[1]{\left\lfloor#1\right\rfloor}

\newcommand{\sub}{\subseteq}
\newcommand{\ria}{\rightarrow}
\newcommand{\CC}{\mathcal{C}}
\newcommand{\CS}{\mathcal{S}}
\newcommand{\CR}{\mathcal{R}}

\newcommand{\CV}{\mathcal{V}}
\newcommand{\CT}{\mathcal{T}}
\newcommand{\NN}{\mathbb{N}}
\newcommand{\RR}{\mathbb{R}}
\newcommand{\EE}{\mathbb{E}}
\newcommand{\PP}{\mathbb{P}}

\newcommand{\Core}{\textup{Core}}

\newtheorem{theorem}{Theorem}
\newtheorem{proposition}[theorem]{Proposition}
\newtheorem{lemma}[theorem]{Lemma}

\newtheorem{problem}[theorem]{Open Problem}
\newtheorem{corollary}[theorem]{Corollary}

\declaretheorem[style=definition,sibling=theorem,qed=$\blacksquare$]{definition}
\declaretheorem[style=definition,sibling=theorem,qed=$\blacksquare$,name=Example]{example}

\usepackage{authblk}
\title{Computing the proportional veto core}
\usepackage{svg}
\newcommand{\orcid}[1]{\href{https://orcid.org/#1}{\includeinkscape[inkscapeformat=pdf,width=10pt]{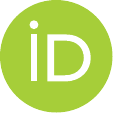}}}

\author{Egor Ianovski\orcid{0000-0001-9411-2529}}
\author{Aleksei Y. Kondratev\orcid{0000-0002-8424-8198}}
\affil{
    HSE University, Saint Petersburg, Russia \\
    george.ianovski@gmail.com, kondratev.aleksei@yandex.ru
}

\begin{document}

\maketitle

\begin{abstract}
In social choice there often arises a conflict between the majority principle (the search for a candidate that is as good as possible for as many voters as possible), and the protection of minority rights (choosing a candidate that is not overly bad for particular individuals or groups). In a context where the latter is our main concern, veto-based rules -- giving individuals or groups the ability to strike off certain candidates from the list -- are a natural and effective way of ensuring that no minority is left with an outcome they find untenable. However, such rules often fail to be anonymous, or impose specific restrictions on the number of voters and candidates. These issues can be addressed by considering the proportional veto core -- the solution to a cooperative game where every coalition is given the power to veto a number of candidates proportional to its size. However, the na\"ive algorithm for the veto core is exponential, and the only known rule for selecting from the core, with an arbitrary number of voters, fails anonymity. In this paper we present a polynomial time algorithm for computing the core, study its expected size, and present an anonymous rule for selecting a candidate from it. We study the properties of core-consistent voting rules. Finally, we show that a pessimist can manipulate the core in polynomial time, while an optimist cannot manipulate it at all.
\end{abstract}

\thispagestyle{empty}

\section{Introduction}

%\edit{Also, excellent motivational examples can be found in book of Moulin 2003. }

\iffalse
\edit{Introductory paragraph goes here}

\citet{Mueller1999} writes:
\begin{quote}
    A voting rule cannot create a political culture, but it can foster compromise and consensus. Both a high qualified majority rule and voting by veto have this property.
\end{quote}
It is the latter that is of interest to us in this paper.

The procedure of \emph{voting by veto} was studied by \citet{Mueller1978}
\fi

Suppose we have a society of 100 individuals who express the following preferences over alternatives $a,b,$ and $c$:
\begin{itemize}
    \item 60 individuals report $a\succ b\succ c$.
    \item 40 individuals report $b\succ c\succ a$.
\end{itemize}
Which alternative should the society adopt? It is reasonable to exclude $c$ from the running, as it is Pareto-dominated by $b$, but both $a$ and $b$ have good arguments to support them, and it is difficult to choose one over another without knowing what this society is and what they are electing. In a political context a lot can be said in favour of $a$. As a president, $b$'s initiatives would be constantly challenged by 60\% of the public, and in a parliamentary system prime minister $b$ would find it very difficult to achieve anything at all. On the other hand, if $a,b$, and $c$ are meeting times, then presumably most people would be able to attend $b$, while with $a$ the hall will be half-empty. As government budgets, $a$ will write off the interests of nearly half the population. As an arbitrated settlement to an armed dispute, $a$ is unlikely to stop the violence. 

Should we decide that in our context the minority matters,
we need to study the profile from the bottom up. Rather than asking how many individuals would love to see $a$ elected, the real issue is to how many such an outcome would be unacceptable. A line of research starting with \citet{Mueller1978} considered voting rules that endowed groups or individuals with the power to unilaterally block certain outcomes from being elected, regardless of how desirable these outcomes may be to the rest of the population.

\citet{Mueller1978} introduced the procedure of \emph{voting by veto} to select a public good, with the explicit aim of protecting an individual from unfair treatment. The alternatives consist of one proposal from each voter, as well as the status quo. An order is formed over the voters and each, in turn, strikes off the alternative they like the least. Exactly one outcome is selected by this procedure, and clearly no voter will see their worst outcome elected. But this procedure is not anonymous, and requires that the number of alternatives be one greater than the voters.

%Of course a procedure like this invites strategic behaviour, and this has been the focus of subsequent work on the subject. The sincere outcome of voting by veto has been shown to correspond to a strong equilibrium solution \citep[example 5]{Moulin1982}, dominance solution \citep{Moulin1980}, and maxmin behaviour \citep{Moulin1981a}; truth-telling may not be an equilibrium, but among equally sophisticated voters the same alternative can be expected to be elected.

\citet{Moulin1981} extended the concept of voting by veto from individuals to coalitions, and studied the core of the resulting cooperative game. The core is anonymous and is well defined for any number of alternatives, but in general this concept is problematic as the core could be empty. However \citeauthor{Moulin1981} has shown that endowing each coalition with the power to veto a number of alternatives proportional to the coalition's size guarantees that the core is non-empty, and is the smallest possible core out of all such rules. \citeauthor{Moulin1981} also proposed \emph{voting by veto tokens} for selecting a specific outcome from this core, however while the core as a concept is anonymous, voting by veto tokens returns to the problem of different voting orders leading to different outcomes.

In the case of two voters, the veto core corresponds to the Pareto-optimal candidates ranked in the top half of both voters' ballots. It turns out that nearly every rule used in this literature selects from the core (for recent developments see \citet{BarberaCoelho22},
\citet{BolLaslierNunez22},
\citet{LaslierNunezSanver21} and the references therein). However, many of these rules are not anonymous, and it is not clear whether they can be extended to an arbitrary number of voters in a reasonable way.

In recent years, voting by veto has been rediscovered in a computational setting. %\footnote{In addition to methods of voting by veto, the computational setting also suggests new applications. For example, group recommender systems \citep{Masthoff15} represent a situation where a small number of agents seek to make a choice from a large pool of outcomes, ideally without fracturing the group. Giving the agents veto rights could be a solution.} 
\citet{Bouveret2017} study the sequential elimination rule, which corresponds to voting by integer veto in the earlier literature \citep[p.~135]{Moulin1983book}.  However their motivation is markedly different from \citet{Mueller1978} -- they emphasise the low communication complexity of the rule, rather than the guarantees it gives to minorities. Indeed, the main focus of the paper is finding an elimination order that best approximates the Borda winner -- and the Borda rule selects the alternative that maximises the average rank, regardless of how many individuals end up with their worst outcome. \citet{KondratevNesterov2020} return to the original motivation, and study the veto core in the context of balancing the rights of minorities and majorities. Their main criticism of the rule is that it is unclear how the core is to be computed -- the na\"ive algorithm is exponential.

\subsection{Our contribution}
We demonstrate the existence of a polynomial-time algorithm for the veto core (\autoref{theorem:corepolynomial}), and study the expected size of the core under Impartial Culture (\autoref{prop:sizecore}, \autoref{prop:sizecore2}) We propose veto by consumption, an anonymous and neutral rule to select from the core, as well as families of rules based on restricting existing voting/ranking rules on candidates known to be in the core (Section~\ref{sec:selecting}), and study the properties of these methods of selecting a winner   (Section~\ref{sec:properties}). Finally, we demonstrate that in the framework of Duggan-Schwartz \citep{Duggan1992,Duggan2000} a pessimist can manipulate the core in polynomial time (\autoref{theorem:pessimist}), while an optimist cannot manipulate it at all (\autoref{thm:strategyproofoptimist}). The conference version of this paper appeared in \citet{IanovskiKondratev21}.

\section{Preliminaries}

We operate in the standard voting model. A voting situation consists of a set $\CV$ of $n$ \emph{voters}, a set $\CC$ of $m$ \emph{candidates}. We assume $\CV$ and $\CC$ are drawn from some well-ordered set such as the natural numbers; in particular, it is always possible to order $\CV$ and $\CC$ from first to last. Every voter is associated with a linear order over the candidates, which we term the voter's \emph{preferences}. We use $\succ_i$ to denote the preference order of voter $i$. When brevity demands it, we may refer to a voter's preferences as his \emph{type}, and use $abc$ to denote the type of a voter with preference order $a\succ b\succ c$. An $n$-tuple of preferences/types is called a \emph{profile}. A mapping that takes a profile to one or more candidates is called a \emph{voting rule}.

A voting rule $\varphi$ is \emph{anonymous} if $\varphi(P)=\varphi(\pi P)$ for each permutation of the voters $\pi$ and each profile~$P$. It is \emph{neutral} if $\varphi(\tau P) = \tau \varphi(P)$ for each permutation of the candidates $\tau$ and each~$P$.

In this paper we are interested in voting rules based on the principle of veto power.

\begin{definition}[\citealp{Moulin1981}]\label{def:vetocore}
A \emph{veto function} is a mapping $v:2^\CV\ria\NN$. Intuitively, $v(T)$ is the number of candidates a coalition of voters $T$ can veto. We call $v(T)$ the \emph{veto power} of $T$.

The \emph{proportional veto function} is the veto function given by:
$$v(T)=\ceil{m\frac{|T|}{n}}-1.$$
A candidate $c$ is \emph{blocked} by a coalition $T$ if there exists a \emph{blocking set} of candidates, $B$, such that:
\begin{align}
    &\forall b\in B,\forall i\in T:\ b\succ_ic,\\
    &m-|B|\leq v(T)
\end{align}
Intuitively, condition (1) means that every voter in $T$ considers every candidate in $B$ to be better than $c$, and condition (2) means that the coalition $T$ can guarantee that the winner will be among $B$ by vetoing all the other candidates.

The set of all candidates that are not blocked with the proportional veto function is called the \emph{veto core}.\footnote{Strictly speaking this is the \emph{proportional veto core}, since a veto core is defined for any veto function. However, outside of \autoref{app:tokens}, the proportional veto core is the only veto core we will consider in this paper, so we will use the shorter term. Within \autoref{app:tokens}, we will explicitly refer to each veto core by its full name.} We shall use $\Core(P)$ to denote the veto core of $P$.
\end{definition}

\begin{theorem}[{\citealp[p. 411]{Moulin1981}}]
For any profile $P$, $\Core(P)$ is non-empty.
\end{theorem}

Note that while the veto core is formulated as a solution concept to a cooperative game, it is also a mapping that takes a profile of preferences to a non-empty subset of candidates -- in other words, it is also a voting rule. As a voting rule it is clearly anonymous and neutral; other properties of the veto core we will explore in Section~\ref{sec:selecting}.

\begin{example}
Consider a profile with five candidates and four voters with the following preferences:
\begin{itemize}
    \item $e\succ_1 b\succ_1 c\succ_1 d\succ_1 a$.
    \item $b\succ_2 e\succ_2 c\succ_2 d\succ_2 a$.
    \item $d\succ_3 b\succ_3 e\succ_3 c\succ_3 a$.
    \item $a\succ_4 c\succ_4 d\succ_4 e\succ_4 b$.
\end{itemize}
In the case of $m=n+1$ the veto power of a coalition simplifies to $\ceil{(n+1)\frac{k}{n}}-1=k$. In other words, $k$ voters can veto exactly $k$ candidates.

Candidate $a$ is blocked by the singleton coalition $\set{1}$ (among others), with the blocking set $B=\set{b,c,d,e}$; $b$ is blocked by $\set{4}$, with $B=\set{a,c,d,e}$; $d$ is not blocked by any singletons, but is blocked by $\set{1,2}$, with $B=\set{b,c,e}$; $c$ is blocked by the coalition $\set{1,2,3}$, with $B=\set{b,e}$. Thus the unique candidate in the veto core is~$e$.

Now let us add a fifth voter:
\begin{itemize}
    \item $c\succ_5 a\succ_5 d\succ_5 b\succ_5 e$.
\end{itemize}
With $m=n$, the veto power simplifies to $\ceil{n\frac{k}{n}}-1=k-1$. A singleton coalition can no longer block anything, but the coalition $\set{1,2}$ can block $a$ with $B=\set{b,c,d,e}$. However, no other candidate is blocked, so the veto core is $\set{b,c,d,e}$.
\end{example}

\section{Computing the core}

The veto core is a solution to a cooperative game, and a na\"ive approach would have us enumerate all possible coalitions. In this section we will show that it is possible to reduce this problem to detecting a biclique in a bipartite graph. To do this, we will need a lemma of \citet{Moulin1981}, and a standard result of algorithmics. Omitted proofs can be found in \autoref{app:proofs} unless stated otherwise.

\begin{restatable}{lemma}{moulin}\textup{(\citealp[p. 409]{Moulin1981})}\label{lem:moulin}
Let $T$ be a coalition of size $k$, and $r,t$ be coefficients satisfying $rn=tm-\alpha$, where $\alpha=\gcd(m,n)$ is the greatest common divisor of $m$ and $n$, and $t>\alpha n$. The veto power of $T$ then satisfies:
$$v(T)=\ceil{m\frac{k}{n}}-1=\floor{\frac{rk}{t}}.$$
\end{restatable}

\begin{restatable}{proposition}{bicliquesize}\textup{(\citealp[p. 196]{Garey1979book})}\label{prop:biclique}
Let $G$ be a bipartite graph, and $k$ an integer. We can, in polynomial time, determine whether there exists a biclique $K_{x,y}\sub G$ with $x$ vertices on the left and $y$ vertices on the right, such that $x+y\geq k$.
\end{restatable}

\begin{samepage}
\begin{theorem}\label{theorem:corepolynomial}
The veto core can be computed in polynomial time.
\end{theorem}
\begin{proof}
We shall show that we can check whether a fixed candidate $c$ is blocked in polynomial time -- we can then find the veto core by running the algorithm below on every candidate.

Consider a profile with $n$ voters and $m$ candidates. Using the Extended Euclidean algorithm we can find $r',t'$ such that $r'n=t'm-\alpha$, where $\alpha=\gcd(m,n)$, and $|t'|\leq\frac{n}{\alpha}$ in $O(\min(\log m,\log n))$ time. It is easy to see that with $(r,t)=(r'+3\alpha m,t'+3\alpha n)$, we have $rn=tm-\alpha$ and $t\geq 3\alpha n-\frac{n}{\alpha}>\alpha n$. By Lemma~\ref{lem:moulin}, the veto power of a coalition of size $k$ in this instance is $\floor{\frac{rk}{t}}$.

Construct a bipartite graph $G$ with $rn$ vertices on the left and $t(m-1)$ vertices on the right. Every voter is associated with $r$ vertices on the left and every candidate but $c$ with $t$ vertices on the right. Add an edge between every voter vertex and every candidate vertex that voter considers to be better than~$c$. Call this the \emph{blocking graph}. We claim that there exists a biclique $K_{x,y}\subseteq G$ with $x+y\geq tm$ if and only if $c$ is blocked.

Suppose that $c$ is blocked by coalition $T$ of size $k$. This means there exists a $B$, $|B|= m-\floor{\frac{rk}{t}}$, such that each voter in $T$ prefers everything in $B$ to $c$. Observe that the vertices associated with $T$ and $B$ form a biclique with $rk + tm - t\floor{\frac{rk}{t}}$ vertices. Since $\frac{rk}{t}\geq\floor{\frac{rk}{t}}$, it follows that $rk\geq t\floor{\frac{rk}{t}}$ and the biclique has at least $tm$ vertices.

Now suppose there exists a biclique $K_{x,y}\subseteq G$ with $x+y\geq tm$. Without loss of generality, we can assume $x=rk',y=tb'$ (since if voter $i$ considers a candidate $a$ to be better than $c$, then all $r$ vertices associated with $i$ are adjacent to all $t$ vertices associated with $a$). That is, there are $k'$ voters who all agree that $b'$ candidates are better than $c$. We must show that the coalition has enough veto power to force this outcome: $\floor{\frac{rk'}{t}}\geq m-b'$. We know that:
\begin{align*}
rk'+tb'&\geq tm\\
rk'&\geq tm-tb'\\
 \frac{rk'}{t}&\geq m-b'.
\end{align*}
Since $m-b'$ is an integer it follows that $\floor{\frac{rk'}{t}}\geq m-b'$.

By Proposition~\ref{prop:biclique} we can, in polynomial time, determine whether such a biclique exists. This is sufficient to establish the problem is solvable in polynomial time. However, in order to use the argument in Proposition~\ref{prop:biclique} we need to find the size of the maximum matching in the complement of the blocking graph, $\overline{G}$. If we use a standard $O(\text{Edges}\sqrt{\text{Vertices}})$ solution our algorithm will run in $O(\alpha^2 m^3n^2\sqrt{\alpha mn})$ time -- a heroic $O(n^{8.5})$ when $m=n=\alpha$. We can do better by reformulating the problem as one of maximum flow.

Define the \emph{flow graph} of $c$ to be the directed weighted graph with a source node $S$, a sink node $U$, a node for every voter, and a node for every candidate but $c$. There is an arc from $S$ to each voter with capacity $r$, an arc from each candidate node to $U$ with capacity $t$, and an arc of unbounded capacity between every voter node and every candidate node that the voter considers to be worse than~$c$.

We will show that if $\overline{G}$ has a matching of size $F$ then the flow graph has a flow of at least $F$, and if the flow graph has a maximum flow of $F$ then $\overline{G}$ has a matching of size at least $F$. Combining the two, the size of the maximum matching in $\overline{G}$ is equal to the maximum flow in the flow graph.

Suppose $\overline{G}$ has a matching $M$ of size $F$. For each edge $(i',d')$ in $M$, where $i'$ is a vertex of voter $i$ and $d'$ is a vertex of candidate $d$, increase the flow in $(S,i)$, $(i,d)$, and $(d,U)$ by one. Since there are at most $r$ and $t$ edges in the matching incident on a vertex of $i$ and $d$ respectively, we have enough capacity in $(S,i)$ and $(d,U)$ to perform this operation for every edge in $M$, and thus construct a flow of size $F$.

Now suppose there is a maximum flow of $F$ in the flow graph. By the Integral Flow Theorem, we can assume that this flow is integral. For every unit of flow from $i$ to $d$, match one vertex of $i$ in $\overline{G}$ with one vertex of $d$. Since the inflow of $i$ is at most $r$, we have enough vertices of $i$ in $\overline{G}$, and since the outflow of $d$ is at most $t$, we have enough vertices of $d$.

Since we can solve the maximum flow problem in cubic time \citep{Malhotra1978}, this gives us an $O(m\max(n^3,m^3))$ solution to the veto core.
\end{proof}
\end{samepage}

\section{Selecting from the core}\label{sec:selecting}

\subsection{The size of the core}

Proportional veto power guarantees that the veto core is non-empty. Moreover, any anonymous veto function with a non-empty core must include the proportional veto core \citep[p. 411]{Moulin1981}; it is thus the smallest solution possible. In general, however, it is not a singleton. In fact, with the Impartial Culture assumption, where all preference orders are generated uniformly and independently, the veto core tends to be unreasonably large (Table~\ref{tab:coresize}). Indeed, in the case of a large number of voters and a small number of candidates, we can demonstrate that the veto core will contain half the candidates on average.

\begin{restatable}{proposition}{ICcoresize}\label{prop:sizecore}
With the Impartial Culture assumption and a fixed number of candidates $m\geq 2$, as $n\ria\infty$, the following are true:
\begin{enumerate}
    \item With probability approaching one, the veto core consists of the candidates which are ranked last by less than $\frac{n}{m}$ voters;
    \item The expected size of the veto core goes to half the candidates.
\end{enumerate}
\end{restatable}

The case of a fixed number of voters and a large number of candidates may not seem as striking in \autoref{tab:coresize}, but here, too, the size of the veto core is unbounded.

\begin{restatable}{proposition}{ICcoresizefixedn}\label{prop:sizecore2}
With the Impartial Culture assumption and a fixed number of voters $n\geq 2$, as $m\ria\infty$, the expected size of the veto core goes to infinity.
\end{restatable}

\begin{table*}
\centering
\begin{tabular}{r | c c c c c c c c c c c c c}
& $m=2$ & 3 & 4 & 5 & 6 & 7 & 8 & 9 & 10& 11& 100& 101& 200\\
\hline
$n=2$ & 0.75 & 0.39 & 0.41 & 0.28 & 0.28 & 0.22 & 0.23 & 0.20 & 0.20& 0.17
&0.04& 0.04& 0.02\\
3 & 0.50 & 0.69 & 0.35 & 0.33 & 0.40 & 0.28 & 0.27 & 0.31 & 0.25& 0.24
&0.09& 0.09& 0.06\\
4 & 0.68 & 0.45 & 0.66 & 0.35 & 0.37 & 0.35 & 0.42 & 0.30 & 0.32& 0.31
&0.15& 0.15& 0.11\\
5 & 0.50 & 0.42 & 0.47 & 0.66 & 0.36 & 0.37 & 0.36 & 0.39 & 0.44& 0.31
&0.21& 0.20& 0.17\\
6 & 0.65 & 0.63 & 0.48 & 0.50 & 0.67 & 0.35 & 0.39 & 0.41 & 0.40& 0.41
&0.26& 0.26& 0.22\\
7 & 0.50 & 0.48 & 0.41 & 0.47 & 0.54 & 0.68 & 0.35 & 0.39 & 0.41& 0.39
&0.29& 0.29 & 0.27\\
8 & 0.63 & 0.44 & 0.64 & 0.46 & 0.52 & 0.58 & 0.70 & 0.36 & 0.40& 0.41
&0.32& 0.32& 0.30\\
9 & 0.50 & 0.63 & 0.53 & 0.41 & 0.52 & 0.54 & 0.60 & 0.71 & 0.36& 0.40
&0.33& 0.33& 0.32\\
10 & 0.62 & 0.51 & 0.51 & 0.65 & 0.47 & 0.50 & 0.57 & 0.63 & 0.71& 0.36
&0.37& 0.34& 0.35\\
11 & 0.50 & 0.45 & 0.44 & 0.58 & 0.42 & 0.49 & 0.54 & 0.60 & 0.64 & 0.71
&0.35 & 0.35 & 0.34\\
100 & 0.54 & 0.53 & 0.55 & 0.57 & 0.49 & 0.54 & 0.52 & 0.57 & 0.58& 0.57
&0.73& 0.37& 0.55\\
101 & 0.50 & 0.50 & 0.53 & 0.53 & 0.48 & 0.52 & 0.50 & 0.55 & 0.57& 0.56
&0.73& 0.74& 0.55\\
200 & 0.52 & 0.49 & 0.54 & 0.54 & 0.52 & 0.50 & 0.55 & 0.54 & 0.56& 0.54
&0.68& 0.41& 0.73\\
\end{tabular}
\caption{Proportion of candidates in the veto core, average of 1,000 IC profiles.}\label{tab:coresize}
\end{table*}

\subsection{The world of core-consistency}
%\subsection{Common properties of core-consistent rules}

We have seen above that the veto core can be large, and thus the space of voting rules that select from the veto core is vast. In the remainder of this section we will explore this family of voting rules.

We first wish to argue that this family is something fundamentally novel, and the ethical principles incorporated in core-consistency are not simple corollaries of, and in fact are often mutually exclusive with, standard desiderata commonly studied in voting theory. As a first step we will show that the class of core-consistent voting rules is disjoint from two of the most famous principles in voting: the utilitarian principle of selecting the candidate with the highest overall ``quality'', usually traced to Borda and Laplace; and the majoritarian principle of respecting the will of the majority of the voters, associated with the work of Condorcet \citep[part~II]{Black58book}. A~comprehensive overview of voting rules and their properties can be found in \citet{FelsenthalNurmi18book}.

\begin{definition}\label{def:propertiesvoting}
A candidate $c$ is a \emph{majority winner} in a profile $P$ if it is ranked first by over half the voters. 

A voting rule $\varphi$ is \emph{majority-consistent} if for each $P$ where $c$ is a majority winner we have $\varphi(P)=c$. This includes the class of Condorcet methods as well as plurality, single transferable vote (STV), and Bucklin.

%We use both in many Propositions below
A candidate $a$ \emph{permutedly dominates} $b$ in $P$ if $a$ has at least as many first positions, first + second positions, first + second + third positions, etc, as $b$, and at least one of these inequalities is strict \citep[p.~166]{Fishburn71,Fishburn73book}. If all of these inequalities are strict, then $a$ \emph{strictly permutedly dominates}~$b$ \citep{Fishburn82}.

%We use positional-consistency only in Proposition 11, we never use strong positional consistency for voting rules
A voting rule $\varphi$ is \emph{positional-consistent} if for each $P$ where $a$ permutedly dominates $b$ and $b\in \varphi(P)$ we have $a\in \varphi(P)$. Positional-consistent rules include rank dependent scoring rules \citep{Goldsmith14}, and median-based rules such as Bucklin and convex median \citep{KondratevNesterov2020}.

%Bucklin, and another variant of the median-based rule \citep{KondratevNesterov2020}.

A voting rule $\varphi$ is \emph{core-consistent} if $\varphi(P)\subseteq\Core(P)$ for each~$P$. By our definition the veto core itself is a core-consistent voting rule, albeit one with an unreasonably high rate of ties. Other examples we shall meet later in this section.
\end{definition}

\begin{proposition}\label{CoreCondorcetDisjoint}
Core-consistent and majority-consistent voting rules are disjoint.
\end{proposition}
\begin{proof}
Consider the profile:
\begin{itemize}
\item $a\succ b\succ c\succ d$.
    \item $a\succ b\succ c\succ d$.
    \item $b\succ c\succ d\succ a$. 
\end{itemize}
Here, $a$ is the majority winner, but the veto core is $\set{b}$.
\end{proof}

\begin{proposition}\label{CorePositionalDisjoint}
Core-consistent and positional-consistent voting rules are disjoint.
\end{proposition}
\begin{proof}
Consider the profile:
\begin{itemize}
    \item $a\succ d\succ e\succ f\succ b\succ c$.
    \item $a\succ c\succ e\succ f\succ b\succ d$.
    \item $a\succ c\succ d\succ f\succ b\succ e$.
    \item $b\succ c\succ e\succ d\succ a\succ f$.
    \item $b\succ c\succ e\succ d\succ a\succ f$.
\end{itemize}
Observe that $a$ and $b$ are always ranked either first or fifth, $a$ has 3 first places, $b$ has 2, $a$ has 5 first and fifth places, $b$ has 5. Hence, $b$ should not be a unique winner under a positional-consistent rule. However, $c,d,e,f$ are blocked by singletons and $a$ is blocked by the last two voters, so the veto core is $\set{b}$.
\end{proof}

If we consider utilitarianism and majoritarianism important, then this is not a good start, but we can console ourselves that every core-consistent rule at least satisfies weaker notions of efficiency and majority.

\begin{definition}
A candidate $c$ is a \emph{majority loser} in a profile~$P$ if it is ranked last by more than half the voters. A voting rule $\varphi$ satisfies the \emph{majority loser criterion} if for each $P$ where $c$ is a majority loser we have $c\notin \varphi(P)$.

A candidate $c$ \emph{Pareto-dominates} $a$ in $P$ if $c\succ_i a$ for every voter $i$. If a candidate is undominated in~$P$, we say that it is \emph{Pareto-optimal}. A voting rule is \emph{Pareto-efficient} if every election winner is a Pareto-optimal candidate.
\end{definition}

\begin{proposition}
Every core-consistent voting rule satisfies:
\begin{enumerate}
    \item Majority loser criterion (indeed, such a rule will never elect a candidate ranked last by more than $n/m$ voters).
    \item Pareto-efficiency.
\end{enumerate}
\end{proposition}
\begin{proof}
\begin{enumerate}
    \item Majority loser criterion.
    
    Suppose $c$ is ranked last by more than $n/m$ the voters. The coalition that ranks $c$ last has a veto power of at least $\ceil{m\frac{\floor{n/m} + 1}{n}} - 1\geq\ceil{1 + \epsilon} - 1\geq 1$, so at the very least they can guarantee that $c$ is blocked.
    
    \item Pareto efficiency.
    
    The grand coalition will block every Pareto-dominated candidate, so the veto core consists only of Pareto-optimal candidates.
\end{enumerate}
\end{proof}

Unfortunately, there are a lot of other properties that are inconsistent with core-consistency.

\begin{definition}
For a profile $P$, a set $C$ of two or more candidates is a set of clones if $a,c\in C$ and $a\succ_i b\succ_i c$ for a voter $i$ imply that $b\in C$. Denote $P_{-c}$ the profile restricted to all candidates but~$c$. A voting rule $\varphi$ satisfies \emph{independence of clones} if for each $P$, $C$, $c\in C$, and $a\notin C$, removing clones does not change the election outcome \citep{Tideman1987}:
\begin{align*}
    a\in \varphi(P)\quad&\text{if and only if}\quad a\in \varphi(P_{-c});\\
    C\cap \varphi(P)\neq\emptyset\quad&\text{if and only if}\quad 
    \set{C\setminus c}\cap \varphi(P_{-c})\neq\emptyset.
\end{align*}

Recall that a candidate $c$ is Pareto-dominated by $a$ in a profile~$P$ if $a\succ_i c$ for every voter~$i$. A~voting rule satisfies \emph{independence of Pareto-dominated alternatives} if removing Pareto-dominated candidates does not change the election outcome \citep[p.~288]{LuceRaiffa57book}.\footnote{Independence of Pareto-dominated alternatives is called reduction principle/condition by \citet[p.~148]{Fishburn71,Fishburn73book}.}

Let $P'$ and $P''$ be profiles with disjoint voters such that for a voting rule $\varphi$, $\varphi(P')\cap \varphi(P'')\neq\emptyset$. Let $P$ be obtained by adding the voters in $P'$ to $P''$. We say that $\varphi$ satisfies \emph{electoral consistency} if $\varphi(P)=\varphi(P')\cap \varphi(P'')$ \citep{Young75}.

Let $P$ be a profile and $P'$ be obtained by reversing the preference order of every voter in~$P$. A~voting rule $\varphi$ satisfies \emph{duality} if $\varphi(P)\cap \varphi(P')$ is either empty or consists of all candidates \citep[p.~157]{Fishburn73book}.%\footnote{In the case of only two candidates, duality is equivalent to neutrality by \citet{May52}.}

A \emph{Condorcet loser} in $P$ is a candidate $c$ such that for any other candidate $a$, the majority of voters prefer $a$ to $c$. A voting rule $\varphi$ satisfies the \emph{Condorcet loser criterion} if for each $P$ where $c$ is a Condorcet loser we have $c\notin \varphi(P)$.
\end{definition}

\begin{proposition}
No core-consistent voting rule satisfies:
\begin{enumerate}
    \item Independence of clones.
    \item Independence of Pareto-dominated alternatives.
    \item Electoral consistency.
    \item Duality.
    \item Condorcet loser criterion.
\end{enumerate}
\end{proposition}
\begin{proof}

\begin{enumerate}
    \item Independence of clones.
    
    Consider the profiles:
    
    \begin{tabular}{ccc}
         \parbox{0.3\textwidth}{
         \begin{itemize}
             \item $a\succ b$.
             \item $b\succ a$.
         \end{itemize}}&  
         \parbox{0.3\textwidth}{
         \begin{itemize}
             \item $a\succ b\succ c$.
             \item $b\succ c\succ a$.
         \end{itemize}}
         &  
         \parbox{0.3\textwidth}{
         \begin{itemize}
             \item $a\succ c\succ b$.
             \item $b\succ a\succ c$.
         \end{itemize}}
    \end{tabular}
    On the left, the veto core is $\set{a,b}$. In the middle, $\set{b}$. On the right, $\set{a}$. Suppose the winner on the left is unique. Without loss of generality, let it be $a$. However, the profile in the middle is obtained by introducing a clone of $b$ ($c$), and the winner is now $b$. 
    
    Now suppose the outcome on the left is the tie $\set{a,b}$. Again, going to the profile in the middle we introduce a clone and change the outcome of the election.
\item Independence of Pareto-dominated alternatives.

In the above example, $c$ is also a Pareto-dominated candidate.

\item Electoral consistency

Consider the following two profiles:

\begin{tabular}{ll}
     \parbox{0.4\textwidth}{\begin{itemize}
         \item $c\succ a\succ b\succ d$.
         \item $c\succ a\succ b\succ d$.
         \item $c\succ a\succ d\succ b$.
         \item $c\succ a\succ d\succ b$.
         \item $b\succ d\succ a\succ c$.
     \end{itemize}}
     &  \parbox{0.4\textwidth}{
     \begin{itemize}
         \item $a\succ c\succ b\succ d$.
         \item $a\succ c\succ b\succ d$.
         \item $a\succ c\succ d\succ b$.
         \item $a\succ c\succ d\succ b$.
         \item $d\succ c\succ b\succ a$.
         \item $\mathbf{d\succ c\succ b\succ a}$.
         \item $\mathbf{b\succ d\succ a\succ c}$.
     \end{itemize}}
\end{tabular}

On the left, a coalition of size four has veto power three, so the four voters ranking $c$ first guarantee the veto core is $\set{c}$. On the right, a coalition of size two has veto power one, so $d,b,a$ are blocked and the veto core is $\set{c}$. In the combined profile, any core-consistent, electorally-consistent rule must choose only $c$.

Now consider the following two profiles, obtained from the previous example by moving the voters in boldface left:

\begin{tabular}{ll}
     \parbox{0.4\textwidth}{\begin{itemize}
         \item $c\succ a\succ b\succ d$.
         \item $c\succ a\succ b\succ d$.
         \item $c\succ a\succ d\succ b$.
         \item $c\succ a\succ d\succ b$.
         \item $b\succ d\succ a\succ c$.
         \item $\mathbf{d\succ c\succ b\succ a}$.
         \item $\mathbf{b\succ d\succ a\succ c}$.
     \end{itemize}}
     &  \parbox{0.4\textwidth}{
     \begin{itemize}
         \item $a\succ c\succ b\succ d$.
         \item $a\succ c\succ b\succ d$.
         \item $a\succ c\succ d\succ b$.
         \item $a\succ c\succ d\succ b$.
         \item $d\succ c\succ b\succ a$.
     \end{itemize}}
\end{tabular}

On the right, the four voters ranking $a$ first guarantee the veto core is $\set{a}$. On the left, two voters each block $d,b,c$ and the veto core is $\set{a}$. Any core-consistent, electorally-consistent rule must choose only $a$ from the combined profile, but that is the same profile as before.

\item Duality.

Consider the profile:
\begin{itemize}
    \item $a\succ b\succ c$.
    \item $c\succ b\succ a$.
\end{itemize}
The veto core is $\set{b}$ in both the profile and the reversed profile.
\item Condorcet loser criterion.

Consider the profile:
\begin{itemize}
    \item $a\succ b\succ c\succ d$.
    \item $b\succ d\succ c\succ a$.
    \item $d\succ a\succ c\succ b$.
\end{itemize}
The veto core is $\set{c}$, but $c$ is the Condorcet loser.
\end{enumerate}
\end{proof}

Independence of clones and independence of Pareto-dominated alternatives are very desirable properties since they formalise resistance to various forms of agenda-manipulation such as vote-splitting. Unfortunately, very few anonymous and neutral rules have these properties: ranked pairs \citep{Tideman1987}, split cycle \citep{HollidayPacuit23pc}, Schulze's method \citep{Schulze11}, and STV. However, the first item of the above proof also shows that every rule satisfying anonymity, neutrality, and positive-responsiveness (defined below) fails both independence properties, so this is an inevitable trade-off. Because of this we consider a weaker notion of independence in the following sections: independence of unanimous losers.

Electoral consistency is an obviously desirable feature, and its violation seems like a violation of common sense.\footnote{For most desirable properties in social choice theory it is easy to come up with a situation where violating that property makes sense, but the authors have been unable to come up with a reasonable violation of electoral consistency. This property can be incompatible with other desirable properties, but such incompatibilities cannot imply that electoral consistency is undesirable in itself. We invite the readers to send us any suggestions.} However, no Condorcet method satisfies this property (\citealp[theorem~2]{YoungLevenglick78}, and \citealp[proposition~2.5]{Zwicker16}), and indeed the only anonymous and neutral electorally consistent voting rules are scoring rules \citep[theorem~1]{Young75}; this is a property we have to do without if we subscribe to any view of voting but the utilitarian. We shall see later that core-consistency is consistent with weaker versions of this principle, such as negative and positive involvement.

The failure of the other two properties is less concerning. The Condorcet loser criterion is fundamentally majoritarian, while core-consistency by design is about protecting minorities from the majority. The ethical principles behind duality are opaque to say the least: $b$ seems a natural compromise candidate in the profile consisting of $abc$ and $cba$, or failing that $\set{a,c}$ seems like a reasonable tie, but the only anonymous and neutral solution admitted by duality is $\set{a,b,c}$.

Thus far we have talked about properties shared by all core-consistent voting rules. In the next section we will consider concrete families of this class.

\subsection{Core-consistent voting rules}
%\subsection{Families of core-consistent rules}

Given that we now know how to compute the veto core in polynomial time, two ad-hoc solutions for selecting from the veto core readily present themselves.

\begin{definition}
Given a voting rule $\varphi$, define $\CS_\varphi$ to be the voting rule that, on profile $P$, first computes $\Core(P)$, then computes $\varphi$ on the reduced profile:
$$\CS_\varphi(P)=\varphi(P\mid_{\Core(P)}).$$

A mapping that takes a profile to one or more rankings of the candidates from best to worst is called a \emph{ranking rule}. Given ranking $R$, let $\text{top}(R)$ be the top ranked element of $R$. Given a ranking rule $\rho$, define $\CR_\rho$ to be the voting rule that, on profile $P$, computes $\rho(P)$, then selects the top ranked candidate(s) in $\Core(P)$ with respect to rankings in $\rho(P)$:
$$\CR_\rho(P)=\bigcup_{R\in \rho(P)}\text{top}(R\mid_{\Core(P)}).$$
\end{definition}

Ranking and voting rules are closely related, and the properties of voting rules can be formalised for ranking rules as well.

\begin{definition}
A ranking rule $\rho$ is \emph{anonymous} if $\rho(P)=\rho(\pi P)$ for each permutation of the voters~$\pi$ and each profile~$P$. It is \emph{neutral} if $\rho(\tau P)=\tau \rho(P)$ for each permutation of the candidates~$\tau$ and each~$P$.

%We use it in Propositions 38 and 44

Let $\underline{R}(c)$ be the set of candidates ranked no higher than $c$ under~$R$. A ranking rule $\rho$ is \emph{positional-consistent} if for each $P$ where $a$ permutedly dominates $b$ (\autoref{def:propertiesvoting}) we have that for each $R\in \rho(P)$ there exists $R'\in \rho(P)$ with $\underline{R}(b)\subseteq \underline{R'}(a)$. It is \emph{strongly positional-consistent} if for each $P$ where $a$ permutedly dominates $b$ we have that $a$ is ranked higher than $b$ in each $R\in \rho(P)$.

%We use it in Propositions 38 and 45
Let $\overline{R}(c)$ be the set of candidates ranked no lower than $c$ under~$R$. A ranking rule $\rho$ is \emph{iterative positional-consistent} if for each triple of profile $P$, ranking $R\in \rho(P)$, and candidate~$c$, this candidate does not strictly permutedly dominate any other candidate in the restricted profile~$P\mid_{\overline{R}(c)}$.

%We use it in Proposition 39
Given a ranking rule $\rho$, define the voting rule $\text{top}(\rho)$ by $\text{top}(\rho)(P)=\bigcup_{R\in \rho(P)} \set{\text{top}(R)}$. Let $X$ be any property defined for voting rules. We say that $\rho$ satisfies $X$ \emph{for top-ranked candidates} if $\text{top}(\rho)$ satisfies $X$.
\end{definition}

%\textcolor{blue}{Aleksei: Note that all S(f) and all other core-consistent choice rules C belong to the family R(r): just construct the general ranking R such that all chosen by rule C candidates on top. R(r) is thus a useful characterisation that helps to generate core-consistent rules from well-known ranking rules.} 

Clearly both $\CS_\varphi$ and $\CR_\rho$ select from the veto core, and are neutral and anonymous if the corresponding $\varphi$ and $\rho$ are. However one may well find the underlying logic unsatisfying -- they are core-consistent in the same sense that Black's rule is a Condorcet method \citep[chapter~IX]{Black58book}. 

More in line with the logic of core-consistency is the solution proposed by \citet{Moulin1981}:

\begin{definition}\label{def:vetotokens}
Let $r:\NN\times\NN\ria\NN$ and $t:\NN\times\NN\ria\NN$ be functions such that $r(m,n),t(m,n)$ satisfy the provisos of Lemma~\ref{lem:moulin}, and $L$ be a function such that $L(k,n)$ returns a linear order over $kn$ elements. 

\emph{Voting by veto tokens} (with parameters $r,t,L$) is the voting rule where, on profile $P$, we create $t(m,n)$ clones of every candidate and endow every voter with $r(m,n)$ veto tokens. $L(r(m,n),n)$ defines an order over the veto tokens and the voters successively veto their least-preferred candidate in the order specified by $L$ until $\gcd(m,n)$ clones are left. The winners are all candidates who have at least one clone remaining.
\end{definition}

\begin{example}\label{ex:tokens}
Let $r,t$ return the smallest integers satisfying the provisos of Lemma~\ref{lem:moulin}. That is, letting $\alpha=\gcd(m,n)$, $r$ and $t$ are the following functions:
\begin{align*}
    t(m,n) &=\min\set{t'\mid\exists r'\text{ such that } r'n = t'm-\alpha \;\text{and}\; t'>\alpha n},\\
    r(m,n) &= \frac{t(m,n)m-\alpha}{n}.
\end{align*}
Let $L(r(m,n),n)$ be the order that places all $r(m,n)$ tokens of voter 1 first, followed by all $r(m,n)$ tokens of voter 2, and so on.

Consider the following profile:
\begin{itemize}
\item $c\succ_1 a\succ_1 b$.
\item $a\succ_2 b\succ_2 c$.
    \item $b\succ_3 a\succ_3 c$.
\end{itemize}
With $n=m=3$, $\alpha=3$, 10 is the smallest integer greater than $\alpha n=9$ for which there exists an $r'$ satisfying $3r'=3\cdot 10 - 3 =27$, and thus $r(3,3)=9$ and $t(3,3)=10$. $L(9,3)$ is the order consisting of 9 tokens of voter 1, followed by 9 tokens of voter 2, followed by 9 tokens of voter 3.

Make 10 clones of every candidate. Voter 1 uses his tokens to veto 9 clones of $b$. Voter 2 vetoes 9 clones of $c$. Voter 3 vetoes the last clone of  $c$, and 8 clones of $a$. At the end of the procedure $a$ has 2 clones remaining and $b$ has 1, so the outcome is a tie between $a$ and $b$.
\end{example}

While different choices of $r$, $t$, and $L$ may lead to different outcomes, such an outcome will always lie in the veto core.

\begin{restatable}{proposition}{vetobytokens}\textup{(\citealp[p. 413]{Moulin1981})}\label{prop:vetobytokens}
For every choice of $r,t,L$, voting by veto tokens is a core-consistent voting rule.
\end{restatable}

Indeed we can show a stronger property, but the proof is tedious (see \autoref{app:tokens}).

\begin{restatable}{proposition}{tokenscore}\label{lem:tokenscore}
Suppose candidate $c$ is in the veto core of profile~$P$. For every choice of $r,t$, there exists an $L$ such that voting by veto tokens elects~$c$.
\end{restatable}

This gives us a family of neutral, core-consistent rules. However the procedure is inherently non-anonymous. If we are willing to admit a randomised procedure we could, of course, simply choose an order at random; if we are to stick to the standard deterministic model, we can propose a simultaneous variant, inspired by  the probabilistic
serial mechanism of \citet{BogomolnaiaMoulin01}:

\begin{definition}
\emph{Veto by consumption} is the voting rule that is computed by an algorithm that has voters eat the candidates from the bottom of their order up. Every candidate starts with capacity~1, and is being eaten by the voters who rank it last. Each voter eats at speed~1.

The outcome can be computed as follows. In round $k$, let $c_i$ be the capacity of candidate $i$ and $n_i$ the number of voters eating $i$. The round lasts until some candidate is fully eaten. To move to round $k+1$, do the following:
\begin{enumerate}
    \item Find an $i$ which minimises $c_i/n_i$. Let $r_k$ be this minimum ratio -- this is the duration of the round.
    \item Update all capacities, $c_j=c_j-r_kn_j$.
    \item For all candidates who reached capacity 0, reallocate the voters eating these to their next worst candidates.
\end{enumerate}
The last candidate to be eaten is the winner. In the case of two or more candidates being eaten simultaneously, a tie is declared among those candidates.
\end{definition}

\begin{restatable}{proposition}{vetobyconsumption}
%Veto by consumption selects a candidate in the veto core.
Veto by consumption is a core-consistent voting rule.
\end{restatable}

Observe that the duration of the algorithm is exactly $\frac{m}{n}$ units, and the complexity is $O(mn)$ arithmetic operations.\footnote{As the size of the numbers involved may grow in every round, we cannot say that the \emph{total} complexity is $O(mn)$.} As a voting rule, veto by consumption is clearly anonymous, neutral, Pareto-efficient, and never selects the candidates which are ranked last by more than $\frac{n}{m}$ voters. 

\section{Properties of core-consistent voting rules}\label{sec:properties}

\subsection{Monotonicity}

Monotonicity is meant to capture the notion that increased support for a candidate should not harm that candidate.

\begin{definition}
%Let $P'$ be obtained from $P$ by having one voter raise $c$ one position in his voting order. A voting rule $\varphi$ is \emph{monotonic} if whenever $c\in \varphi(P)$, it is also the case that $c\in \varphi(P')$.

Let $P'$ be obtained from profile $P$ by having one voter raise candidate $c$ one position in his voting order. A voting rule $\varphi$ is \emph{monotonic} if for each $P$ where $c\in \varphi(P)$ we have $c\in \varphi(P')$.

Let $\underline{R}(c)$ be the set of candidates ranked no higher than $c$ under~$R$. A ranking rule $\rho$ is \emph{monotonic} if for each $R\in \rho(P)$ there exists $R'\in \rho(P')$ such that $\underline{R}(c)\subseteq \underline{R'}(c)$.

\end{definition}

Monotonicity is a very basic property satisfied by many voting rules. The exceptions are typically iterative rules such as STV, plurality with run-off, or sequential majority elimination \citep[theorem~2]{Fishburn82,Smith73}, which can lead to controversy since such rules are used in many real elections.

It turns out that it is quite easy to combine monotonicity and core-consistency: it is satisfied by the veto core, voting by veto tokens, veto by consumption, and $\CR_\rho$ for suitable choice of $\rho$. Since $\CS_\varphi$ is morally an iterative rule, it should be not surprising that it  almost always fails monotonicity.

\begin{restatable}{proposition}{vetocoremonotone}
    The veto core is monotonic.
\end{restatable}

\begin{proposition}
Let $\varphi$ be a voting rule that is majority-consistent when the number of candidates is two. Then $\CS_\varphi$ is not monotonic.
\end{proposition}
\begin{proof}
Consider a profile with one voter each of types $cba$, $bac$, $acb$, and two voters each of types $abc$, $cab$, $bca$. The veto core is $\set{a,b,c}$. Observe that $a$ is preferred by the majority to $b$, $b$ to $c$, and $c$ to $a$.

Without loss of generality, suppose that $a$ is a (possibly tied) winner. Now suppose the voter of type $cba$ changes their preferences to $cab$, moving $a$ up one position. The veto core is now $\set{a,c}$, since $b$ is ranked last by four voters. However $c$ is the unique winner, since $\varphi$ is the majority rule for $m=2$.
\end{proof}

Observe how extremely weak the requirements on $\varphi$ are -- almost every well-known voting rule reduces to the majority rule in the case of two candidates.
However we cannot relax the requirements further -- if $\varphi$ is trivial (i.e., $\CS_\varphi$ is equivalent to the veto core), imposed (according to some fixed ranking of candidates), or dictatorial, then $\CS_\varphi$ is monotonic.

\begin{restatable}{proposition}{rrmonotone}\label{prop:rankmonotonic}
If a ranking rule $\rho$ is monotonic, then $\CR_\rho$ is a monotonic voting rule.
\end{restatable}

The proof that veto by consumption and voting by veto tokens are monotonic is a bit involved, and can be found in \autoref{app:monotonic}.

\begin{proposition}\label{prop:HHAmonotonic}
Veto by consumption and voting by veto tokens are monotonic.
\end{proposition}

\subsection{Positive responsiveness}

Positive responsiveness is a strengthening of monotonicity that captures the idea that the smallest extra support for a tied candidate is enough to break the tie.

\begin{definition}
Let $P'$ be obtained from profile $P$ by having one voter raise candidate $c$ one position in his voting order. A voting rule $\varphi$ satisfies \emph{positive responsiveness} if for each $P$ where $c\in \varphi(P)$ we have $\varphi(P')=c$.

Let $\underline{R}(c)$ be the set of candidates ranked no higher than $c$ under $R$. A ranking rule $\rho$ is positively responsive if $\underline{R}(c)\subseteq \underline{R'}(c)$ for all $R\in \rho(P)$ and $R'\in \rho(P')$.
\end{definition}

Positive responsiveness is harder to satisfy than monotonicity. It is intuitively a positional notion, and thus it is not surprising that strongly monotonic scoring rules rules are positive responsive, but it is failed by many majority-consistent rules such as Copeland, Simpson's maxmin rule, and STV. We see below that core-consistency is consistent with positive responsiveness, but it is achieved through what is essentially a positive responsive tie-breaking mechanism (\autoref{prop:rrposresponse}, \autoref{prop:HHAGAposresponse}).

\begin{proposition}\label{prop:sfposresponse}
Let $\varphi$ be a neutral and anonymous voting rule. Then $\CS_\varphi$ is not positively responsive.
\end{proposition}
\begin{proof}
Consider the profile with a voter of type $dabc$ and a voter of type $cbad$. The veto core is $\set{a,b}$, and $f$ must declare a tie because in the reduced profile we have one voter of type $ab$ and one of type $ba$. After changing the type of the first voter to $adbc$ the veto core is still $\set{a,b}$, and the reduced profile is still $ab$, $ba$, so the candidates are tied, which violates positive responsiveness.
\end{proof}

We cannot relax assumptions. If $\varphi$ is imposed (according to some fixed ranking of candidates), then $\varphi$ is anonymous and $\CS_\varphi$ is positively responsive. If $\varphi$ is dictatorial, then $\varphi$ is neutral and $\CS_\varphi$ is positively responsive. 

By \autoref{prop:sfposresponse}, the veto core fails positive responsiveness, because the $\varphi$ that elects every candidate is anonymous and neutral, and $\CS_\varphi=\Core$.

\begin{restatable}{proposition}{rrpositiveresponsive}\label{prop:rrposresponse}
If a ranking rule $\rho$ is positively responsive, then $\CR_\rho$ is a positively responsive voting rule.
\end{restatable}

\begin{proposition}
Veto by consumption fails positive responsiveness.
\end{proposition}
\begin{proof}
Consider the profile with a voter of type $abcd$ and $cdab$. The winners are $\set{a,c}$. Change the type of the second voter to $cadb$. The winners are still $\set{a,c}$.
\end{proof}

\subsection{Independence of unanimous losers}

A unanimous loser is a candidate that is ranked last by every single voter. Independence of unanimous losers formalises the notion that such candidates should have no effect on the election.

\begin{definition}
A candidate is a \emph{unanimous loser} if the candidate is ranked last in every voter's preference order. A voting rule satisfies \emph{independence of unanimous losers} if removing the unanimous loser does not change the set of elected outcomes.
\end{definition}

The practical relevance of such independence is protection from spoilers -- it is relatively easy to add a terrible candidate to an election, and we should hope that the agenda is not that easy to manipulate. By itself independence of unanimous losers is a fairly weak axiom, yet it is violated by a number of well-known rules such as antiplurality (without tie-breaking) and Nanson's procedure  \citep{KondratevIanovskiNesterov19}. 

Veto by consumption clearly satisfies independence of unanimous losers, since all these candidates are eaten simultaneously by all the voters, before moving on to the real candidates. On the other hand it is easy to see that the veto core fails this property -- once we add $m(n-1)$ unanimous losers,
$$v(T)=\ceil{mn\frac{|T|}{n}}-1=m|T|-1,$$
and only the grand coalition will have enough veto power to block any real candidate. The veto core will consist of all Pareto-optimal candidates. 

The other core-consistent rules fare far worse.

\begin{proposition}
For every voting rule $\varphi$, $\CS_\varphi$ fails independence of unanimous losers.
\end{proposition}
\begin{proof}
Suppose, for contradiction, that $\CS_\varphi$ satisfies independence of unanimous losers. Consider the profile with voter 1 of type $bca$ and voter 2 of type $abc$. The veto core is $\set{b}$, so $b$ is the unique winner. Let $P$ be the profile with voter 1 of type $bcad$ and voter 2 of type $abcd$. It is obtained by adding the unanimous loser $d$, so by independence of unanimous losers $\CS_\varphi(P)=b$. Consider another profile with voter 1 of type $bac$ and voter 2 of type $acb$. The veto core is $\set{a}$, so $a$ is the unique winner. Let $P'$ be the profile with voter 1 of type $bacd$ and voter 2 of type $acbd$, which is obtained by adding the unanimous loser $d$. It must be the case that $\CS_\varphi(P')=a$, but that is a contradiction because $\Core(P)=\Core(P')=\set{a,b}$ and the two profiles coincide on the relative ranking of $a$ and $b$, so $\varphi$ must choose the same winner in both cases.
\end{proof}

\begin{proposition}
If $\rho$ is majority-consistent for top-ranked candidates, positional-consistent, or iterative positional-consistent, then $\CR_\rho$ fails independence of unanimous losers.
\end{proposition}
\begin{proof}
Suppose $\rho$ is majority-consistent for top-ranked candidates. Consider the profile with three voters of type $abc$ and two voters of type $cba$. The veto core is $\set{b}$, so $b$ is the winner. Now add the unanimous loser $d$ to obtain three voters of type $abcd$ and two of type $cbad$. The veto core is $\set{a,b}$, $a$ is the majority winner, and hence $\CR_\rho$ must select only~$a$. 

Assume, for contradiction, that $\rho$ is positional-consistent and satisfies independence of unanimous losers. Recall the profile from \autoref{CorePositionalDisjoint}, and add the unanimous loser $g$.
\begin{itemize}
    \item $a\succ d\succ e\succ f\succ b\succ c\succ g$.
    \item $a\succ c\succ e\succ f\succ b\succ d\succ g$.
    \item $a\succ c\succ d\succ f\succ b\succ e\succ g$.
    \item $b\succ c\succ e\succ d\succ a\succ f\succ g$.
    \item $b\succ c\succ e\succ d\succ a\succ f\succ g$.
\end{itemize}
The veto core was $\set{b}$ before the addition of $g$. However, the veto core is now $\Core(P)=\set{a,b,c,d,e}$, and $a$ permutedly dominates~$b$. By independence of unanimous losers, for each $R\in \rho(P)$, $b=\text{top}(R\mid_{\Core(P)})$ and hence $\Core(P)\subseteq \underline{R}(b)$. Because $\rho$ is positional-consistent, for each $R\in \rho(P)$, there exists some $R'\in \rho(P)$ such that $\underline{R}(b)\subseteq\underline{R'}(a)$. Then $\Core(P)\subseteq \underline{R'}(a)$, $a=\text{top}(R'\mid_{\Core(P)})$ and hence $a$ also wins.

Suppose $\rho$ is iterative positional-consistent. Consider the profile with 2 voters of type $acb$, 4 of type $bca$, 5 of type $abc$. The veto core is $\set{b}$. After adding the unanimous loser $d$, the veto core is $\set{a,b}$. For each $R\in \rho(P)$, $d$ is ranked last, because all other candidates strictly permutedly dominate~$d$. Then $c$ is ranked previous to last, because both $a$ and $b$ strictly permutedly dominate $c$ in the profile restricted to $\set{a,b,c}$. Then $b$ is ranked below $a$, because $a$ strictly permutedly dominates $b$ in the profile restricted to $\set{a,b}$. Hence, $a$ is the unique winner.
\end{proof}

The above proposition shows that $\CR_\rho$ fails independence of unanimous loser for many reasonable choices of $\rho$. We cannot prove such a proposition for \emph{all} choice of $\rho$, since veto by consumption satisfies independence of unanimous losers, so we could choose $\rho$ to be the elimination order of candidates in a run of veto by consumption. However, as the following proposition shows, this is essentially all we can do: if $\rho$ satisfies some minimal requirements but is not itself core-consistent, then $\CR_\rho$ fails independence of unanimous losers.

\begin{proposition}
Let $\rho$ be a ranking rule that satisfies independence of unanimous losers and Pareto efficiency, but is not core-consistent (all three properties for top-ranked candidates). Then $\CR_\rho$ fails independence of unanimous losers.
\end{proposition}
\begin{proof}
Suppose candidate $c$ is top-ranked by some ranking in $\rho(P)$, Pareto-optimal, but not in the veto core. Hence $c$ is not selected by $\CR_\rho(P)$. In the profile $P'$ that results from $P$ by adding $m(n-1)$ unanimous losers, the veto core consists of all Pareto-optimal candidates, including~$c$. Because $\rho$ satisfies independence of unanimous losers for top-ranked candidates, $c$ is still top-ranked by some ranking in $\rho(P')$ and hence is selected by $\CR_\rho(P')$.
\end{proof}

%\subsubsection{Participation}
\subsection{Participation}

Voter adaptability, positive  involvement, and negative involvement are variations of the participation principle that a voter should not regret turning up to the polling booth. The properties formalised by these axioms are conceptually similar, but none of the three imply another.

\begin{definition}
For a profile $P$, we denote $P_{-i}$ the profile obtained by removing a voter~$i$ from~$P$. 

A voting rule $\varphi$ satisfies \emph{voter adaptability} if for each profile~$P$ where voter~$i$ ranks a candidate~$c$ first and $\varphi(P_{-i})=c$ we have $\varphi(P)=c$ \citep{Richelson78III}.

A voting rule $\varphi$ satisfies \emph{positive involvement} if for each profile~$P$ where voter~$i$ ranks a candidate~$c$ first and $c\in \varphi(P_{-i})$ we have $c\in \varphi(P)$ \citep[proposition~2(iii)]{Smith73}.

A voting rule $\varphi$ satisfies \emph{negative involvement} if for each profile~$P$ where voter~$i$ ranks a candidate~$c$ last and $c\in \varphi(P)$ we have $c\in \varphi(P_{-i})$ \citep[p.~93]{Meredith1913book}.\footnote{If a voting rule satisfies negative involvement, then it is resistant to the no-show paradox in \citet{FishburnBrams83}.} Rule~$\varphi$ satisfies \emph{strong negative involvement} if in addition $\varphi(P_{-i})=c$.
\end{definition}

Participation is usually defined to require $i$ to prefer the outcome with him than without him, $\varphi(P)\succeq_i\varphi(P_{-i})$, but considerable differences arise as to how authors deal with ties. It is a very strong property and among well-known rules it is only satisfied by the scoring rules, and is failed by every iterative scoring rule (\citealp[corollary~6.1]{Saari89jet}, and \citealp[theorem~7]{NunezSanver17}) and Condorcet method \citep{Moulin1988}.\footnote{\citet[corollary~6.1]{Saari89jet} considers resolute (i.e. single-valued) election outcomes and participation of two voters of the same type. \citet[theorem~7]{NunezSanver17} assume that ties are broken lexicographically. \citet{Moulin1988} and \citet[theorems~3,4]{BrandtGeistPeters17} show that every resolute Condorcet method fails participation. \citet[proposition~2]{JimenoPerezGarcia09} and \citet[theorems~5-8]{BrandtGeistPeters17} show that every irresolute Condorcet method fails participation for optimists or pessimists.} The three properties we consider are much weaker than standard participation, and rules satisfying voter adaptability, positive or negative involvement are somewhat more common. Scoring rules satisfy all three, as do certain Condorcet methods such as Simpson's maxmin rule, whereas Young's rule satisfies negative involvement but not positive involvement or voter adaptability \citep{Perez01} and split cycle satisfies positive and negative involvement \citep{HollidayPacuit23pc}, but not voter adaptability. The Coombs rule satisfies negative involvement but not positive involvement or voter adaptability \citep[p.~60]{FelsenthalNurmi18book}, plurality with runoff and STV satisfy voter adaptability, positive but not negative involvement (\citealp[p.~93]{Meredith1913book}, and \citealp[p.~208]{FishburnBrams83}).\footnote{STV is the only single-winner iterative scoring rule (with one by one elimination) that satisfies positive involvement \citep[p.~258]{Saari94book}. However, \citet[p.~212]{FishburnBrams83} show that STV does not satisfy positive involvement as a multiwinner voting rule and acknowledge that it has been known at least since 1910.}

%BrandtMatthausSaile22 find minimal counterexamples for the Young, Coombs, plurality with runoff and STV: Young's rule satisfies negative involvement but not positive involvement or voter adaptability. The Coombs rule satisfies negative involvement but not positive involvement or voter adaptability, plurality with runoff and STV satisfy voter adaptability, positive but not negative involvement

We have been unable to devise a core-consistent rule that satisfies voter adaptability, but neither could we show that core-consistency and voter adaptability are inconsistent; the problem remains open. Negative involvement seems easier for a core-consistent rule to satisfy than positive involvement, but the veto core satisfies both.

\begin{restatable}{proposition}{vetocoreinvolvement}\label{prop:coreinvolvement}
The veto core satisfies both positive  involvement and negative involvement, but fails strong negative involvement.
\end{restatable}

However the veto core fails voter adaptability, as does every rule from the $\CS_\varphi$ family.

\begin{proposition}\label{prop:sfadaptability}
For every voting rule $\varphi$, $\CS_\varphi$ fails voter adaptability.
\end{proposition}
\begin{proof}
Consider two profiles:

\begin{tabular}{ll}
     \parbox{0.4\textwidth}{\begin{itemize}
         \item $a\succ_1b\succ_1c.$
         \item $a\succ_2b\succ_2c.$
         \item $a\succ_3b\succ_3c.$
         \item $b\succ_4c\succ_4a.$
         \item $b\succ_5c\succ_5a.$
         \item $b\succ_6a\succ_6c.$
     \end{itemize} }
     &
     \parbox{0.4\textwidth}{\begin{itemize}
         \item $a\succ_1b\succ_1c.$
         \item $a\succ_2c\succ_2b.$
         \item $a\succ_3c\succ_3b.$
         \item $b\succ_4a\succ_4c.$
         \item $b\succ_5a\succ_5c.$
         \item $b\succ_6a\succ_6c.$ 
     \end{itemize} }
\end{tabular}

Consider the profile consisting of voters 1--5 on the left. The veto core is $\set{b}$. Once we add 6, the veto core becomes $\set{a,b}$, and by voter adaptability $b$ should remain the unique winner.

Now consider the profile consisting of voters 2--6 on the right. The veto core is $\set{a}$. Once we add 1, the veto core becomes $\set{a,b}$, and by voter adaptability $a$ should remain the unique winner. However, that is impossible because the left and right profiles coincide on their restriction to $\set{a,b}$.
\end{proof}

We cannot extend the above proposition to positive or negative involvement, because the veto core satisfies these properties, and $\CS_\varphi=\Core$ with the $\varphi$ that elects all candidates. However, imposing a very weak condition on $\varphi$ is sufficient.

\begin{proposition}\label{prop:sfinvolvement}
Let $\varphi$ be a majority-consistent voting rule when the number of candidates is two. Then $\CS_\varphi$ fails positive involvement and negative involvement.%and resolute participation
\end{proposition}
\begin{proof}
Consider profile $P_{-i}$ with five voters of type $abc$ and three of type $bca$. The veto core is $\set{b}$. If we add voter $i$ with type $bac$ the veto core becomes $\set{a,b}$, and every majority-consistent $\varphi$ selects only $a$. Hence, $\CS_\varphi$ fails positive involvement.

Consider a profile $Q$ with four voters (1--4) of type  $abc$, four (5--8) of type $cab$, four (9--12) of type $bca$. Without loss of generality, suppose $a$ is among the winners according to $\varphi$, $a\in \varphi(Q)$. Now let $P$ be a profile with four voters (1--4) of type $abcd$, one (voter 5) $cabd$, three (6--8) $cadb$, three (9--11) $bcad$, and voter $i$ of type $bcda$. The veto core is $\set{a,b,c}$, and the restriction of $P$ to $\set{a,b,c}$ is $Q$, so $a\in\CS_\varphi(P)=\varphi(Q)$. However, if we remove $i$, then the veto core is $\set{a,c}$ in~$P_{-i}$, and only $c$ wins by majority, $\CS_\varphi(P_{-i})=c$, violating negative involvement.
\end{proof}

Whether or not a rule of the family $S_\varphi$ can satisfy strong negative involvement remains open.

\begin{problem}
Does there exist a voting rule $\varphi$ for which $S_\varphi$ satisfies strong negative involvement?
\end{problem}

Veto by consumption, meanwhile, satisfies strong negative involvement, but fails both positive involvement and voter adaptability.

%I constructed the proof such that it works for any tie-breaking, including generalised antiplurality
\begin{restatable}{proposition}{VbCnegativeInvolvement}\label{prop:VbCnegativeInvolvement}
Veto by consumption satisfies strong negative involvement.
\end{restatable}

\begin{proposition}
Veto by consumption fails voter adaptability and positive involvement.%and resolute participation
\end{proposition}
\begin{proof}
Consider the voters:
\begin{itemize}
    \item $a\succ b\succ c\succ d\succ e\succ f\succ g\succ h\succ i\succ j\succ k$.
    \item $k\succ j\succ i\succ h\succ g\succ f\succ e\succ d\succ c\succ b\succ a$.
    \item $f\succ a\succ b\succ c\succ d\succ e\succ g\succ h\succ i\succ j\succ k$.
\end{itemize}
In the profile with the first two voters the winner under any core-consistent voting rule is candidate~$f$. With all three voters the winner is $d$, but the third voter prefers $f$ to any other candidate. (Candidates $a,b,k,j,i,h$ are eaten first. Then $g$ is eaten and the capacity of $c$ is reduced to $1/2$. Then $c$ is eaten and the capacity of $f,e$ are reduced to $1/2$. Then $f,e$ are eaten and the capacity of $d$ is reduced to~$1/2$.)
\end{proof}

The $\CR_\rho$ family fares no better with respect to these properties.

\begin{proposition}\label{prop:rrposinvolvement}
Let $\rho$ be a positional-consistent ranking rule. Then $\CR_\rho$ fails voter adaptability. If, in addition, $\rho$ is strongly positional-consistent, then $\CR_\rho$ fails positive involvement.%and resolute participation
\end{proposition}
\begin{proof}
Suppose $\rho$ is positional-consistent. Consider the profile from \autoref{CorePositionalDisjoint}:
\begin{itemize}
    \item $a\succ d\succ e\succ f\succ b\succ c$.
    \item $a\succ c\succ e\succ f\succ b\succ d$.
    \item $a\succ c\succ d\succ f\succ b\succ e$.
    \item $b\succ c\succ e\succ d\succ a\succ f$.
    \item $b\succ c\succ e\succ d\succ a\succ f$.
\end{itemize}
The veto core is $\set{b}$. Let $P$ be a profile consisting of these five voters and a voter of type $bacdef$. The veto core is now $\set{a,b}$. Observe that $a$ permutedly dominates $b$. If $b$ is a winner under $\CR_\rho$ in $P$, then $b$ is ranked higher than $a$ under some $R\in \rho(P)$. Because $\rho$ is positional-consistent, there exists some $R'\in \rho(P)$ such that $\underline{R}(b)\subseteq\underline{R'}(a)$. Hence, $a$ is ranked higher than $b$ under $R'$ and also wins. So, the outcome is either $a$ or $\set{a,b}$, violating voter adaptability. If $\rho$ is strongly positional-consistent, then $a$ is the unique winner, violating positive involvement.
\end{proof}

\begin{proposition}\label{prop:rrmajinvolvement}
Suppose $\rho$ is majority-consistent for top-ranked candidates, or iterative positional-consistent. Then $\CR_\rho$ fails voter adaptability and positive involvement.%and resolute participation
\end{proposition}
\begin{proof}
Consider a profile with five voters of type $abc$ and three of type $bca$. The veto core is $\set{b}$. If we add a voter of type $bac$, the veto core becomes $\set{a,b}$. The majority ranks $a$ first, so if $\rho$ is majority-consistent for top-ranked candidates then $a$ is the unique winner under~$\CR_\rho$. If $\rho$ is iterative positional-consistent, then $c$ is ranked last, because both $a$ and $b$ strictly permutedly dominate $c$. Then $b$ is ranked below $a$, because $a$ strictly permutedly dominates $b$ when only these two candidates remain. In either case voter adaptability and positive involvement are violated.
\end{proof}

The below proposition may not sound surprising, but it covers almost all  $\rho$ not covered by the previous propositions.

\begin{restatable}{proposition}{rrualinvolvement}\label{prop:rrualinvolvement}
Let $\rho$ satisfy independence of unanimous losers and Pareto efficiency, but fail voter adaptability/positive involvement (all four properties for top-ranked candidates). Then $\CR_\rho$ also fails voter adaptability/positive involvement.
\end{restatable}

Propositions \ref{prop:sfadaptability}, \ref{prop:sfinvolvement}, \ref{prop:rrposinvolvement}, \ref{prop:rrmajinvolvement}, and \ref{prop:rrualinvolvement} cover every way we could think of to select from the veto core using a known voting/ranking rule, and show that such approaches will fail positive involvement and voter adaptability. That does not necessarily mean these properties are inconsistent with core-consistency, but that we have to consider rules that select from the veto core directly. For example, we have seen in \autoref{prop:coreinvolvement} that the veto core satisfies positive involvement. However we have been unable to find a core-consistent voting rule satisfying voter adaptability.

\begin{problem}
Is voter adaptability consistent with core-consistency?
\end{problem}

\subsection{The matter of ties}

We have argued that the veto core is itself a voting rule, albeit with a high rate of ties. All the core-consistent rules we have introduced, therefore, could be viewed merely as tie-breaking mechanisms for the veto core. But what if these tie-breaking mechanisms themselves have ties?

Voting by veto tokens has the property that for coprime $m$ and $n$ exactly one winner is selected ($gcd(m,n)=1$), but veto by consumption offers no such guarantees. Simulation results suggest that the rate of ties is small for large elections (\autoref{tab:hungry}), and indeed we can show that the rule satisfies large voter population resolvability.

\begin{table*}
\centering
\begin{tabular}{r | c c c c c c c c c c c c c }
& $m=2$ & 3 & 4 & 5 & 6 & 7 & 8 & 9 & 10 & 11 & 100 & 101 & 200\\
\hline
$n=2$ & 1.5 & 1.167 & 1.42 & 1.22 & 1.38 & 1.24 & 1.37 & 1.25 & 1.35 & 1.26 & 1.31 & 1.3 & 1.31\\
3 & 1 & 1.6 & 1.04 & 1.02 & 1.41 & 1.04 & 1.03 & 1.33 & 1.04 & 1.03 & 1.04 & 1.04 & 1.04\\
4 & 1.38 & 1.24 & 1.58 & 1.13 & 1.28 & 1.2 & 1.36 & 1.15 & 1.23 & 1.18 & 1.13 & 1.12 & 1.12\\
5 & 1 & 1 & 1 & 1.42 & 1 & 1 & 1 & 1 & 1.21 & 1 & 1.02 & 1 & 1\\
6 & 1.31 & 1.36 & 1.21 & 1.15 & 1.3 & 1.08 & 1.1 & 1.14 & 1.09 & 1.08 & 1.03 & 1.03 & 1.02\\
7 & 1 & 1 & 1 & 1 & 1 & 1.18 & 1 & 1 & 1 & 1 & 1 & 1 & 1\\
8 & 1.27 & 1.08 & 1.24 & 1.08 & 1.11 & 1.06 & 1.11 & 1.04 & 1.05 & 1.03 & 1.01 & 1.01 & 1.01\\
9 & 1 & 1.22 & 1.01 & 1 & 1.08& 1.03 & 1.02 & 1.06 & 1.01 & 1.01 & 1 & 1 & 1\\
10 & 1.25 & 1.14 & 1.11 & 1.13 & 1.05 & 1.03 & 1.02 & 1.01 & 1.04 & 1.01 & 1 & 1 & 1\\
11 & 1 & 1 & 1 & 1 & 1 & 1 & 1 & 1 & 1 & 1.02 & 1 & 1 & 1\\
100 & 1.08 & 1.01 & 1.01 & 1 & 1 & 1 & 1 & 1 & 1 & 1 & 1 & 1 & 1\\
101 & 1 & 1 & 1 & 1 & 1 & 1 & 1 & 1 & 1 & 1 & 1 & 1 & 1\\
200 & 1.06 & 1.01 & 1 & 1 & 1 & 1 & 1 & 1 & 1 & 1 & 1 & 1 & 1\\
\end{tabular}
\caption{Number of veto by consumption winners, average of 1,000,000 IC profiles.}\label{tab:hungry}
\end{table*}

\begin{restatable}{proposition}{HHAWoodall}
    With the Impartial Culture assumption and a fixed number of candidates $m$, as $n\ria\infty$, the expected number of winners under veto by consumption tends to~1.
\end{restatable}

However, the rule does not satisfy  large candidate population resolvability.

\begin{restatable}{proposition}{VbCsizentwo}
    For each profile with $n$ voters and $m$ candidates, the number of veto by consumption winners does not exceed the minimum of $n$ and~$m$. 
    
    With the Impartial Culture assumption and $n=2$, the expected number of veto by consumption winners is:
    \begin{equation*}
        2-1+\frac{1}{2}-\frac{1}{3}+\ldots+\frac{(-1)^m}{m},
    \end{equation*}
    which limits to $2-\ln{2}\approx 1.307$ as $m\rightarrow\infty$.
\end{restatable}

Large candidate population resolvability is less commonly studied than large voter population resolvability because most authors approach voting with a political election in mind. Veto-based methods, however, originate in the context of a small number of sophisticated agents making a choice from a potentially vast pool of possible outcomes \citep{Mueller1978, Moulin1981a}. The case of $n=2$ from the above proposition is interesting in particular since the outcomes of veto by consumption correspond exactly to the subgame perfect equilibria of the Alternate Strike scheme \citep[theorem~1]{Anbarci06} and the outcomes of the cautious voting under gradual vetoes mechanism  \citep{BolLaslierNunez22}. Because of this, it is worth considering how ties under veto by consumption should be broken.

We propose using generalised antiplurality\footnote{This the generalised scoring rule where the winners have the least last-place positions, breaking first-order ties with the number of second-to-last positions, second-order ties with the number of third-to-last positions, and so on.} as a tie-breaking rule. Not only is the idea of looking at the profile from the bottom up in line with the general philosophy of core-consistency, but this rule is also sufficiently powerful to restore large candidate population resolvability to veto by consumption, maintains anonymity and neutrality, and weakly improves on the studied properties satisfied by veto by consumption.\footnote{The proof that veto by consumption with generalised antiplurality tie-breaking satisfies strong negative involvement is identical to \autoref{prop:VbCnegativeInvolvement}.}

\begin{restatable}{proposition}{GAtiebreaking}
With the Impartial Culture assumption and a fixed number of voters $n\neq 2$, as $m\ria\infty$, the expected number of winners under any voting rule with generalised antiplurality tie-breaking tends to~1. 

For $n=2$, the expected number of winners under any Pareto-efficient  voting rule with generalised antiplurality tie-breaking tends to~1.
\end{restatable}

\begin{restatable}{proposition}{HHAGAposresponse}\label{prop:HHAGAposresponse}
Veto by consumption with generalised antiplurality tie-breaking satisfies positive responsiveness.
\end{restatable}

\subsection{Comparison}

We illustrate the results of the previous sections by comparing the properties satisfied by several core-consistent voting rules in \autoref{tab:comparing_voting_rules}. We compare the following rules:
\begin{itemize}
    \item $\Core$: the veto core itself.
    \item $\CS_\text{Iter}$: the iterative core, obtained by finding the veto core of the veto core until a fixed point is reached.
    \item $\CR_\text{GA}$: $\CR_\rho$ where $\rho$ ranks the candidates by their generalised antiplurality score.
    \item $\CR_\text{Lex}$: $\CR_\rho$ where $\rho$ ranks the candidates in lexicographical order. Note that this is identical to~$\CS_\text{Lex}$.
    \item UFT: up-front tokens, voting by veto tokens with the same choice of $r,t,L$ as in \autoref{ex:tokens}.
    \item VbC: Veto by consumption.
    \item $\text{VbC}_\text{GA}$: Veto by consumption with generalised antiplurality tie-breaking.
\end{itemize}
If a result in the table does not follow from the preceding sections, a proof can be found in \autoref{app:specificrules}.

\begin{table}
\caption{Properties of specific core-consistent voting rules.}\label{tab:comparing_voting_rules}
    {\small
    \begin{center}
    \begin{tabular}{l l l l l l l l}
    \toprule
         &  $\Core$ & $\CS_{\text{Iter}}$ & $\CR_{\text{GA}}$ & $\CR_{\text{Lex}}$ &UFT& VbC & $\text{VbC}_\text{GA}$\\
         \midrule
         Anonymity&Yes&Yes&Yes&Yes&No&Yes&Yes\\
         Neutrality&Yes&Yes&Yes&No&Yes&Yes&Yes\\
         Monotonicity&Yes&No&Yes&Yes&Yes&Yes&Yes\\
         Positive responsiveness&No&No&Yes&Yes&No&No&Yes\\        Independence of UL&No&No&No&No&No&Yes&Yes\\            
         Voter adaptability&No&No&No&No&No&No&No\\
         Positive involvement&Yes&No&No&No&No&No&No\\
         Negative involvement&Yes&No&No&No&?&Yes&Yes\\
         Strong negative involvement&No&No&No&No&?&Yes&Yes\\
         %Time complexity&P*&$\textup{P}^\dag$&P*&P*&$\textup{P}^\S$&$\textup{P}^\P$&$\textup{P}^\P$\\
         Expected size, $n\ria\infty$&$m/2$&1&1&1&$\leq m/2$&1&1\\
         Expected size, $m\ria\infty$&$\infty$&$\leq n$&1&1&$\leq n$&$\leq n$& 1\\
         \bottomrule
    \end{tabular}
    \end{center}}
    \vspace{0.2cm}
%\footnotesize{\emph{Notes}: * -- $O(m\max(n^3,m^3))$, \dag -- $O(m^2\max(n^3,m^3))$, \S -- $O(mn\min(m,n))$ ,\P -- $O(mn)$.}
\end{table}

\section{Manipulating the core}

The Gibbard-Satterthwaite framework of strategic voting presupposes a resolute voting rule, i.e.\ one that always outputs a single winner. As no anonymous rule can have this property, the standard approach in studying manipulation is to assume that ties are broken either lexicographically, in favour of the manipulator, or against the manipulator. Such an approach is reasonable if ties are a rare occurrence, but given the tendency of the veto core to be large, in our context such an approach would lead to a tangible deviation from neutrality or anonymity. Instead, we need to evaluate the manipulator's preferences on the outcome sets directly. \citet{Duggan1992, Duggan2000} considered two kinds of preference extensions:\footnote{\citet{Duggan1992, Duggan2000} use a more complicated definition in their paper, but it is equivalent to the pessimist/optimist definition found in \citet{Taylor2002}.} manipulation by an optimist, who evaluates a set by the best element, and a pessimist, who evaluates a set by the worst. While these extensions may not seem very imaginative, they are enough to derive an impossibility result -- any voting rule that is immune to both types of manipulation and includes all singletons in its range is dictatorial.

Conceptually, manipulating the veto core in this framework is simple. The veto core consists of all candidates that are not blocked; therefore, an agent's strategies are limited to either blocking a candidate they would not sincerely, or refusing to block a candidate they would have otherwise. In the case of pessimistic manipulation, this gives us a polynomial time algorithm.

\begin{lemma}\label{lem:topset}
The \emph{top set} of candidate $c$ in order $P_i$ is the set of all candidates ranked above $c$.

If $c$ is blocked in $(P_i,P_{-i})$ with top set $X$, and $X\subseteq Y$, then for any $P_i'$ with top set $Y$, $c$ is blocked in $(P_i',P_{-i})$.
\end{lemma}
\begin{proof}
Suppose $c$ is blocked in $(P_i,P_{-i})$ with top set $X$. Let $P_i'$ be an arbitrary order where the top set of $c$ is $Y$, $X\subseteq Y$.

Let $T$ be a coalition blocking $c$ in $(P_i,P_{-i})$ with blocking set $B$. If $i\notin T$, then $T$ is a blocking coalition regardless of what $i$ votes, and remains a blocking coalition in $(P_i',P_{-i})$. If $i\in T$, then observe that $B\subseteq X\subseteq Y$, so $i$ considers all the candidates in $B$ to be better than $c$. The other voters' preferences are unchanged, so they do too. Thus $T$ is still a blocking coalition in $(P_i',P_{-i})$.
\end{proof}

\begin{theorem}\label{theorem:pessimist}
It can be determined whether the veto core is manipulable by a pessimist in polynomial time.
\end{theorem}
\begin{proof}
Consider a profile $P$ where $c$ is the worst element in the veto core for the manipulator, who will be voter 1 for convenience. Call all the candidates 1 considers to be better than $c$ the good candidates, all the candidates at least as bad (including $c$ itself) the bad candidates.

The problem of pessimistic manipulation is thus finding a $P_1'$ such that every bad candidate in $(P_1',P_{-1})$ is blocked. Call such a $P_1'$ a strategic vote.

Observe that as a consequence of Lemma~\ref{lem:topset}, if $c$ is blocked in $(P_1',P_{-1})$ with top set $X$, then $c$ is blocked in \emph{every} $(P_1^*,P_{-1})$ with top set $X$. We can check whether this is the case in polynomial time by constructing an arbitrary preference order $P_1^*$ with the candidates in $X$ ranked first, then $c$, and running the veto core algorithm on $(P_1^*,P_{-1})$. We will refer to this as \emph{checking whether $c$ is blocked with top set $X$}.

The algorithm functions by progressively filling the ballot from the top down. At step $i$ of the algorithm, let $G_i$ be the set of candidates already on the ballot. To initialise the algorithm, construct a partially filled ballot with all the good candidates ranked at the top in the same order as in $P_1$. Thus $G_0$ is the set of good candidates. To advance from step $i$ to $i+1$ of the algorithm check if there exists a candidate  $b\in \CC\setminus G_i$ that is blocked with top set $G_i$. If so, rank $b$ directly below the candidates in $G_i$ and let $G_{i+1}=G_i\cup\set{b}$. If not, assert that no strategic vote exists.

First, observe that this algorithm will run for $O(m)$ steps, and each step will involve $O(m)$ calls to the algorithm computing veto core, so the algorithm is in P.

Next, we show that if the algorithm returns a vote it is indeed a strategic vote. To demonstrate this we will show that the algorithm has the invariant property that every candidate in $G_i$ is either a good candidate or is blocked. At the beginning of the algorithm, $G_0$ consists solely of the good candidates. After step $i$, we add some candidate $b$ who is blocked with top set $G_i$. We do not change the top set of any candidate in $G_i$, so those that were blocked at step $i$ remain blocked at step $i+1$. Thus if the algorithm returns a vote, all the candidates on the ballot are either good candidates or blocked.

Finally, we show that if there exists a strategic vote $P_1'$, the algorithm will find one.

We first claim that if there exists a strategic vote, there exists a strategic vote with all the good candidates ranked at the top of the ballot in the same order as in $P_1$. Let $P_1'$ be a strategic vote where some good $g$ is ranked behind a bad $b$. Obtain $P_1^*$ by swapping $g$ and $b$. We increase the top set of $b$, so if $b$ was blocked it remains blocked. Thus $P_1^*$ is still strategic, and we can repeat this operation until all the good candidates bubble to the top. Once we have a ballot where all the good candidates are on the top, their internal order does not matter since it will not change the top set of any bad candidate. There is thus no generality lost in assuming that the good candidates are ranked as in $P_1$.

Now we proceed by contradiction. Suppose the algorithm fails to return a vote. This must mean that at step $i$ of the algorithm none of the remaining bad candidates are blocked with top set $G_i$. By the above claim, we can assume that $P_1'$ ranks the set of good candidates $G_0$ at the top in the same order as $P_1$ and thus has bad candidates in the next $i$ positions. Let $X$ be the set of these $i$ bad candidates, and $X'=G_i\setminus G_0$ be the set of candidates ranked in the corresponding positions by our algorithm. Observe that $X\neq X'$ -- otherwise the algorithm could choose the same candidate in the $i$th step as in $P_1'$. Thus there exists a $b\in X, b\notin X'$. Of all such $b$, choose one that is ranked highest in $P_1'$. Let $Y$ be the set of bad candidates (possibly empty) ranked above $b$ in $P_1'$. Since $b$ is the highest ranked bad candidate in $P_1'$ that is not in $X'$, it must be that $Y\sub X'$. But that is a contradiction because if $b$ is blocked with top set $G_0\cup Y$ in $P_1'$, it will certainly be blocked with top set $G_i=G_0\cup X'$ in the algorithm.
\end{proof}

In principle, it is also easy to check whether optimistic manipulation is possible -- by Lemma~\ref{lem:topset}, if it is possible for voter 1 to add some $c$ to the veto core, then it must be possible to add $c$ to the veto core by ranking $c$ first, and the other candidates in any order. However, it turns out that this polynomial time algorithm can be reduced to ``return false''; the veto core can never be manipulated by an optimist.

\begin{definition}
Voter $i$ is a blocker of $c$ if:
\begin{enumerate}
    \item There exists a coalition $T$, $i\in T$, such that $T$ blocks $c$.
    \item For all coalitions $T'$, $T'\setminus\set{i}$ does not block $c$.
\end{enumerate}
\end{definition}

\begin{lemma}\label{lem:blockercore}
Let $i$ be a blocker of $c$. The veto core contains at least one $a$, $a\succ_i c$.
\end{lemma}
\begin{proof}
Recall that voting by veto tokens elects a candidate in the veto core, for any order of the veto tokens (\autoref{prop:vetobytokens}). Thus if there exists an order of the tokens under which $a\succ_i c$ is elected, then $a$ is in the veto core. We will construct such an order.

Since $i$ is a blocker of $c$, no coalition of voters without $i$ can block $c$. Thus in the modified problem $P'$ where $i$ is replaced with a voter that ranks $c$ first (call this voter $i'$), $c$ must be in the veto core. By \Cref{lem:tokenscore}, there exists an order of the tokens in $P'$ such that at the end of the process $c$ has at least one clone remaining. Call this order $L$.

We claim that we can without loss of generality assume that in $L$, the tokens of $i'$ are cast last. To see this, suppose $i'$ casts a token at position $k$ in $L$, and uses it to veto a clone of candidate $d$. Consider what happens when we move this token to the back of $L$ to obtain $L'$. This will not change the behaviour of any voter before $k$, since the set of clones of candidates remaining will not change. Voter $j$ moving at position $k$ will face a situation where $d$ has one more clone than under $L$. Voter $j$ will either veto the same clone they did under $L$ or veto the extra clone of $d$. If they veto $d$, then all subsequent voters will behave the same under $L'$ as they did under $L$, and $c$ will have a clone remaining. Otherwise, $j$ will veto the same clone as they did under $L$, and we can repeat the argument for the next voter. Thus at the end $c$ will have a clone remaining.

Now apply the order $L'$ to the original problem. After the voters other than $i$ have cast their vetoes, $c$ still has a clone remaining. Since $c$ is not in the veto core, by the time $i$ finished casting his vetoes, $c$ will be eliminated. However, the clones of $i$ will not veto any candidate better than $c$ until $c$ is eliminated. Hence the winner must be some $a\succ_i c$.
\end{proof}
\begin{theorem}\label{thm:strategyproofoptimist}
The veto core is strategy-proof for optimists.
\end{theorem}
\begin{proof}
The veto core consists of all candidates that are not blocked, thus the only way $i$ can add a new candidate to the veto core is to change their vote so that some candidate is no longer blocked. Call this candidate $c$. In order to do this, $i$ must be a blocker of $c$ (as otherwise $c$ would be blocked no matter what $i$ does). However, if $i$ is a blocker of $c$, then the veto core already contains something better than $c$ (\autoref{lem:blockercore}).
\end{proof}

\iffalse
\section{Future directions}
From inception, the concept of voting by veto was closely linked with implementation. \citet{Anbarci06} described the subgame perfect equilibria of the Alternate Strike scheme. It turns out that these equilibria correspond exactly to the veto by consumption winners in the case of two voters. Similarly, \citet{LaslierNunezSanver21} proposed a mechanism which implements the two-voter veto core in Nash equilibrium. It is natural to ask whether something similar can be done with an arbitrary number of voters.

Veto by consumption is a novel voting rule, and hence the standard questions (axiomatisation, non-trivial properties, susceptibility to manipulation, etc) have not been asked about it. As an anonymous voting rule which selects from the veto core, it would be interesting to consider how well it functions in settings where the happiness of the minority matters, e.g.\ in group recommender systems \citep{Masthoff15}.
\fi

\section{Acknowledgements}
%Aleksei Kondratev is supported by the Basic Research Program and the ``Priority 2030'' program at HSE University.

Support from the Basic Research Program and the ``Priority 2030'' program at HSE University is gratefully acknowledged.

\appendix

\section{Omitted proofs}\label{app:proofs}
\moulin*

\begin{proof}
From $rn=tm-\alpha$ we obtain: $$\frac{rk}{t}=\frac{km}{n}-\frac{k\alpha}{tn}.$$
Case one: $\frac{km}{n}$ is an integer. In this case:
\begin{align*}
    \floor{\frac{rk}{t}}&=\floor{\frac{km}{n}-\frac{k\alpha}{tn}},\\
    &=\ceil{m\frac{k}{n}}-1.
    \end{align*}
Case two: $\frac{km}{n}=\floor{\frac{km}{n}}+\epsilon, \epsilon\in\set{\frac{1}{n},\dots,\frac{n-1}{n}}$. In this case $\ceil{m\frac{k}{n}}-1=\floor{\frac{km}{n}}$, and:
\begin{align*}
    \floor{\frac{rk}{t}}&=\floor{\frac{km}{n}-\frac{k\alpha}{tn}},\\
    &=\floor{\frac{km}{n}}+\floor{\epsilon -\frac{k\alpha}{tn}},\\
    &\geq\floor{\frac{km}{n}}+\floor{\epsilon -\frac{1}{n}},\\
    &\geq\floor{\frac{km}{n}}.
\end{align*}
The equality follows because it is clear that: $$\floor{\frac{km}{n}-\frac{k\alpha}{tn}}\leq\floor{\frac{km}{n}}.$$
\end{proof}

\bicliquesize*
\begin{proof}
Let $G=(L,R)$. First observe that we can without loss of generality assume that $k>\max(|L|,|R|)$. Suppose otherwise, and $k=\max(|L|,|R|) - c$ for $c\geq 0$. Construct a bipartite graph $G=(L',R')$ consisting of all the vertices and edges of $G$, plus another $c+1$ vertices in $L'$ that are adjacent to every vertex in $R'$, and $c+1$ vertices in $R'$ that are adjacent to every vertex in $L'$. Clearly $G$ contains a biclique $K_{x,y}$ if and only if $G'$ contains a biclique $K_{x+c+1,y+c+1}$, and if we let $k' = k+2c+2$ then $k'=\max(|L|,|R|) + 2c +2 = \max(|L'|,|R'|) + c + 1 >\max(|L'|,|R'|)$.

By K\H{o}nig's theorem, the size of the maximum matching, $z$, in a bipartite graph is equal to the size of the smallest vertex cover. The complement of the smallest vertex cover is the maximum independent set. The size of the maximum independent set is thus $|L|+|R|-z$, and is equal to  the size of the maximum biclique in the bipartite-complement of the graph, provided that the size of the maximum independent set is greater than $\max(|L|,|R|)$ (otherwise, the independent set could consist entirely of one of the partitions of the graph).

Therefore take $G$, construct its bipartite-complement $\overline{G}$ and find the size of the maximum matching, $z$, which is necessarily no greater than $\min(|L|,|R|)$. The size of the maximum independent set in $\overline{G}$ is then $|L|+|R|-z \geq \max(|L|,|R|)$. If the inequality is strict, then the size of the maximum biclique $K_{x,y}$ in $G$ is $x+y=|L|+|R|-z$ and we simply compare this value with~$k$. If $|L|+|R|-z = \max(|L|,|R|)$, then observe that the maximum independent set in $\overline{G}$ is necessarily smaller than $k$. Since a biclique in $G$ corresponds to an independent set in $\overline{G}$, this means $G$ cannot contain a bliclique of size $k$ or larger.
\end{proof}

\ICcoresize*
\begin{proof}
Fix a set of candidates, $B$, $|B|=m-q, 1\leq q\leq m-1$, and fix a candidate~$c$ not in $B$. We will show that the probability of the existence of a coalition blocking $c$ with blocking set $B$ is $\frac{1}{2}$ if $q=1$, and 0 otherwise. This will establish 1 and 2.

Let $p(q,m)$ be the probability that a voter believes every candidate in $B$ is better than $c$. If $q=1$, then $c$ must be ranked last. There are $(m-1)!$ possible preference orders the voter can have, so $p(1,m)=\frac{(m-1)!}{m!}=\frac{1}{m}$.

If $q>1$, then the voter must rank $c$ between positions $m-q+1$ and $m$ (any higher and the candidates in $B$ could not possibly all fit above $c$). There are $(m-1)!$ orders with $c$ ranked last, and less than $(m-1)!$ with $c$ ranked in any of the other $q-1$ positions. Thus there are more than $(m-1)!$, and less than $q(m-1)!$ preference orders that the voter could have. This gives us $\frac{1}{m}<p(q,m)<\frac{q}{m}$.

Recall that the veto power of a coalition of size $k$ is $\ceil{m\frac{k}{n}}-1$. Thus a coalition can veto $q$ candidates if and only if the size of the coalition is strictly larger than $\frac{nq}{m}$. It follows that $B$ is a blocking set for $c$ if and only if more than $\frac{nq}{m}$ voters believe every candidate in $B$ is better than $c$. Hence, by the Central Limit Theorem, the limit of the probability that $B$ is a blocking set for $c$ is:
\begin{equation*}
    1-\lim_{n\rightarrow\infty}{F\left( \frac{\frac{nq}{m}-np(q,m)}{\sqrt{np(q,m)(1-p(q,m))}} \right)},
\end{equation*}
which is equal to $\frac{1}{2}$ for $q=1$ and 0 otherwise.
\end{proof}

\ICcoresizefixedn*
\begin{proof}
Fix the number of voters, $n\geq 2$. For each profile with $m\geq n$ candidates, let $X_m$, $|X_m| = x_m$, be the set of candidates  which every voter ranks in the top $\floor{\frac{m}{n}}$ positions. Recall from \autoref{def:vetocore} it follows that a candidate in $X_m$ can only be blocked by the grand coalition. Let $Y_m\subseteq X_m$, $|Y_m|=y_m$, be the Pareto-optimal candidates in $X_m$. These candidates are not blocked by the grand coalition, and are hence in the veto core. We shall show that $\EE[y_m]\ria\infty$ as $m\ria\infty$.

Note that $x_m$ takes values from $\{0,1,\ldots,\floor{\frac{m}{n}}\}$ and its expectation is
\begin{equation*}
    \EE[x_m]=m\left(\frac{\floor{\frac{m}{n}}}{m}\right)^n.
\end{equation*}
Denote
\begin{equation*}
    a_m=\floor{\frac{m}{n}}\left(\frac{\floor{\frac{m}{n}}}{m}\right)^n.
\end{equation*}
From
\begin{align*}
    \EE[x_m]&=\sum\limits_{k:k\leq a_m}k\PP[x_m=k]+\sum\limits_{k:k>a_m}k\PP[x_m=k],\\ 
    &\leq\sum\limits_{k:k\leq a_m}a_m\PP[x_m=k]+\sum\limits_{k:k>a_m}\floor{\frac{m}{n}}\PP[x_m=k],\\ 
    &=a_m(1-\PP[x_m>a_m])+\floor{\frac{m}{n}}\PP[x_m>a_m],
\end{align*}
we have that 
\begin{align*}
    \PP[x_m>a_m]&\geq\frac{\EE[x_m]-a_m}{\floor{\frac{m}{n}}-a_m},\\
    &\geq\frac{\left(m-\floor{\frac{m}{n}}\right)\left(\frac{\floor{\frac{m}{n}}}{m}\right)^n}{\floor{\frac{m}{n}}},\\
    &\geq\frac{n\floor{\frac{m}{n}}-\floor{\frac{m}{n}}}{\floor{\frac{m}{n}}}\left(\frac{\frac{m}{n}-1}{m}\right)^n,\\
    &\geq(n-1)\left(\frac{1}{n}-\frac{1}{m}\right)^n.
\end{align*}

For a fixed $k\in\NN$, consider all profiles with $x_m=k$. For each such a profile, there exist $(k!)^n$ profiles which obtained from each other by permutations of these $k$ candidates. Let the first voter rank them in the order $b_1\succ b_2\succ\ldots\succ b_k$. Candidate $b_1$ is always Pareto-optimal. Candidate $b_2$ is higher than $b_1$ in the preference of the second voter with frequency $1/2$ and hence $b_2$ is Pareto-optimal with frequency at least $1/2$. Similarly, $b_3$ is Pareto-optimal with frequency at least $1/3$, and so on. Hence, the conditional expectation of $y_m$ is
\begin{equation*}
    \EE[y_m\mid x_m=k]\geq 1+\frac{1}{2}+\ldots+\frac{1}{k}.
\end{equation*}

This implies that
\begin{align*}
\EE[y_m]=\sum\limits_k{\EE[y_m|x_m=k]\cdot\PP[x_m=k]}&\geq
    \sum\limits_{k:k>a_m}{\left(1+\frac{1}{2}+\ldots+\frac{1}{k}\right)\PP[x_m=k]},\\
    &\geq\left(1+\frac{1}{2}+\ldots+\frac{1}{\floor{a_m}}\right)\PP[x_m>a_m],\\
    &\geq\left(1+\frac{1}{2}+\ldots+\frac{1}{\floor{a_m}}\right)(n-1)\left(\frac{1}{n}-\frac{1}{m}\right)^n.
\end{align*}
As $m\ria\infty$, $(n-1)(\frac{1}{n}-\frac{1}{m})^n$ tends to a constant, while $(1+\frac{1}{2}+\ldots+\frac{1}{\floor{a_m}})$ tends to infinity.
\end{proof}

\vetobytokens*
\begin{proof}
Suppose, for contradiction, that voting by veto tokens elects a candidate $c$, but $c$ is blocked by voters $T$, $|T|=k$. This means there must exist a set of candidates $B$, $|B|=m-v(T)$, such that every voter in $T$ believes every candidate in $B$ is better than $c$. Let $W=\CC\setminus B$, and $W_i$ the set that voter $i\in T$ believes to be at least as bad as $c$. Observe that $W_i\subseteq W$. By Lemma~\ref{lem:moulin}, the number of clones of candidates in~$W$ is equal to $t|W|=t v(T)=t\floor{\frac{rk}{t}}$.

Since a clone of $c$ remains by the end of the algorithm, then it means that each $i\in T$ cast all $r$ of their tokens against a clone of a candidate in $W_i\subseteq W$. Yet these $rk$ vetoes have been insufficient to eliminate all the clones of $c\in W$ and hence $rk<t\floor{\frac{rk}{t}}$, which is a contradiction.
\end{proof}

\vetobyconsumption*
\begin{proof}
Suppose, for contradiction, that veto by consumption elects a candidate $c$, but $c$ is blocked by voters $T$, $|T|=k$. This means there must exist a set of candidates $B$, $|B|=m-v(T)=m-\ceil{\frac{mk}{n}}+1$, such that every voter in $T$ believes every candidate in $B$ is better than $c$. For voter $i\in T$, let $W_i$ be the set of candidates $i$ considers to be at least as bad as $c$. Observe that $W_i\subseteq \CC\setminus B$.

Observe that as long as $c$ is present, every agent will eat a candidate they consider to be at least as bad as $c$. Since we have assumed that $c$ is eaten last, it follows that the $k$ voters in $T$ eat only from $W=\bigcup W_i\subseteq \CC\setminus B$ for the duration of the algorithm, and the capacity of $W$ is $|W|\leq \ceil{\frac{mk}{n}}-1<\frac{mk}{n}$. However, it will take these $k$ voters strictly less than $\frac{m}{n}$ units of time to eat all of $W$, but the algorithm runs for exactly $\frac{m}{n}$ units of time, which is a contradiction.
\end{proof}

\vetocoremonotone*
\begin{proof}
    Suppose $c\in\Core(P)$. Let $P'$ be obtained from $P$ by having voter $i$ raise $c$ one position in his voting order.
    
    We prove by contradiction. If $c\notin\Core(P')$, then that must mean there exists a coalition $T$ and a blocking set $B$ such that:
\begin{align*}
    &\forall b\in B,\forall j\in T:\ b P_j' c,\\
    &m-|B|\leq v(T).
\end{align*}
Since in $P$ $c$ does not improve his position vis-\`a-vis any other candidate, $c$ is blocked by $T$ in~$P$.
\end{proof}

\rrmonotone*
\begin{proof}
Consider a candidate $c\in\CR_\rho(P)$. This means that $c$ is top-ranked in $\Core(P)$ by some ranking $R\in \rho(P)$. The set of candidates ranked no higher than $c$, $\underline{R}(c)$, is such that $\Core(P)\subseteq \underline{R}(c)$. Let $P'$ be obtained from $P$ by having a voter raise $c$ one position in his voting order.

We first claim that $\Core(P')\subseteq\Core(P)$. Observe that if $a\notin\Core(P)$, then that must mean there exists a coalition $T$ and a blocking set $B$ such that:
\begin{align*}
    &\forall b\in B,\forall j\in T:\ b P_j a,\\
    &m-|B|\leq v(T).
\end{align*}
Since in $P'$ $a$ does not improve his position vis-\`a-vis any other candidate, $a$ is still blocked by $T$ in~$P'$.

Since the veto core is monotonic, $c\in\Core(P')$. Since $r$ is monotonic, $\underline{R}(c)\subseteq \underline{R'}(c)$ for some ranking $R'\in \rho(P')$. Hence, $$\Core(P')\subseteq\Core(P)\subseteq \underline{R}(c)\subseteq \underline{R'}(c),$$
$c$ is top-ranked in $\Core(P')$ under $R'$ and $c\in\CR_\rho(P')$.
\end{proof}

\rrpositiveresponsive*
\begin{proof}
Suppose candidate $c\in\CR_\rho(P)$. This means that $c$ is top-ranked in $\Core(P)$ by some ranking $R\in \rho(P)$. The set of candidates ranked no higher than $c$, $\underline{R}(c)$, is such that $\Core(P)\subseteq \underline{R}(c)$. Let $P'$ be obtained from $P$ by having a voter raise $c$ one position in his voting order.

As we showed in \autoref{prop:rankmonotonic}, $c\in \Core(P')\subseteq\Core(P)$. Since $\rho$ is positively responsive, $\underline{R}(c)\subseteq \underline{R'}(c)$ for each ranking $R'\in \rho(P')$. Hence, $$\Core(P')\subseteq\Core(P)\subseteq \underline{R}(c)\subseteq \underline{R'}(c),$$
$c$ is top-ranked in $\Core(P')$ by each ranking in $\rho(P')$ and hence the unique winner.
\end{proof}

\vetocoreinvolvement*

\begin{proof}
Let $P$ be a profile with $n$ voters and $P_{-i}$ be obtained from $P$ by removing voter~$i$.

For positive involvement, suppose $c\in\Core(P_{-i})$ and voter $i$ ranks $c$ first. Assume, for contradiction, that $c$ is blocked in $P$.  This means there exists a coalition $T\subseteq\CV$ and blocking set $B$ satisfying:
\begin{align*}
    &\forall b\in B,\forall j\in T:\ b \succ_j c,\\
    &m-|B|\leq\ceil{m\frac{|T|}{n}}-1\leq\ceil{m\frac{|T|}{n-1}}-1.
\end{align*}
Because $c$ is voter $i$'s top candidate, $i\notin T$ and $c$ is blocked in~$P_{-i}$. 

For negative involvement, suppose $c\in\Core(P)$ and voter $i$ ranks $c$ last. For contradiction, suppose $c\notin\Core(P_{-i})$. This means there exists a coalition $T\subseteq\CV\setminus\set{i}$ and blocking set $B$ satisfying:
\begin{align*}
    &\forall b\in B,\forall j\in T:\ b \succ_j c,\\
    &m-|B|\leq\ceil{m\frac{|T|}{n-1}} -1.
\end{align*}
Voter $i$ ranks $c$ last, so $b\succ_i c$ for all $b\in B$. Thus there exists a coalition $T\cup\set{i}$ for which:
\begin{align*}
    &\forall b\in B,\forall j\in T\cup\set{i}:\ b \succ_j c,\\
    &m-|B|\leq\ceil{m\frac{|T|}{n-1}} -1\leq\ceil{m\frac{|T|+1}{n}} -1.
\end{align*}
As such, $T\cup\set{i}$ blocks $c$ in $P$.

To show that the veto core does not satisfy strong negative involvement consider a profile with voters of type $cba, cab, abc, bac$. The core remains $\set{a,b}$ even if we remove the voter that ranks $a$ last. 
\iffalse
Suppose $c$ is in the veto core of $P$. Let $P'$ be obtained from $P$ by adding $i$ that ranks $c$ first. Assume, for contradiction, that $c$ is blocked in $P'$.  This means there exists a coalition $T$ and blocking set $B$ satisfying:
\begin{align*}
    &\forall b\in B,\forall j\in T:\ b \succ_j' c,\\
    &m-|B|\leq\ceil{m\frac{|T|}{n+1}}-1\leq\ceil{m\frac{|T|}{n}}-1.
\end{align*}
Because $c$ is voter $i$'s top candidate, $i\notin T$ and $c$ is blocked in~$P$.

Suppose $c$ is not in the veto core of $P$. That is, there exists a coalition $T$ and blocking set $B$ satisfying:
\begin{align*}
    &\forall b\in B,\forall j\in T:\ b \succ_j c,\\
    &m-|B|\leq\ceil{m\frac{|T|}{n}} -1.
\end{align*}
Construct $P'$ by adding voter $i$ that ranks $c$ last. Let $T'=T\cup\set{i}$. Observe that:
\begin{align*}
    &\forall b\in B,\forall j\in T':\ b \succ_j' c,\\
    &m-|B|\leq \ceil{m\frac{|T|}{n}}-1\leq \ceil{m\frac{|T'|}{n+1}}-1,
\end{align*}
and $c$ remains blocked in~$P'$.

To show that the veto core does not satisfy strong negative involvement consider a profile with voters of type $cba, cab, abc, bac$. The core remains $\set{a,b}$ even if we remove the voter that ranks $a$ last. 
\fi
\end{proof}

\VbCnegativeInvolvement*
\begin{proof}
Let $c(a,t,P)$ denote the capacity of candidate $a$ at time $t$ during the run of veto by consumption on profile $P$, and $N(a,t,P)$ the set of voters eating $a$ at time $t$.

Suppose in a profile $P$ with $m$ candidates and $n+1$ voters a candidate $a$ is a (possibly tied) winner under veto by consumption and voter $i$ ranks $a$ last. This implies that the capacity of $a$ is positive and $i$ belongs to the set of voters eating $a$ until the end of the algorithm: for each time $t<m/(n+1)$, $c(a,t,P)>0$ and $i\in N(a,t,P)$. We will show that $a$ is the unique winner in~$P_{-i}$.

Denote $t_k<m/(n+1)$ the starting time of each round $k$ during the run of the algorithm on~$P$. We will prove by induction on $k$ that for each $b\neq a$, we have
\begin{align}
    c(a,t_k,P_{-i})&\geq c(a,t_k,P),\label{neginvolvement:ainequality}\\ 
    c(b,t_k,P_{-i})&=c(b,t_k,P),\label{neginvolvement:bequality}\\
   N(a,t_k,P_{-i})&= N(a,t_k,P)\setminus{i},\label{neginvolvement:aN}\\ 
   N(b,t_k,P_{-i})&=N(b,t_k,P).\label{neginvolvement:bN}
\end{align}
In the base case of $k=1$, $t_1=0$, all capacities are equal to 1 and the only difference is that $i$ is eating $a$ in $P$ but not in~$P_{-i}$. Suppose the inductive hypothesis holds for round~$k$. Let $r_k>0$ be the duration of round $k$ during the run of the algorithm on $P$, and thus $t_{k+1}=t_k+r_k$. Then:
\begin{align*}
    c(a,t_{k+1},P_{-i})&=c(a,t_k,P_{-i})-r_k|N(a,t_k,P_{-i})|\\
    &=c(a,t_k,P_{-i})-r_k|N(a,t_k,P)|+r_k\\
    &>c(a,t_k,P)-r_k|N(a,t_k,P)|=c(a,t_{k+1},P)\geq 0,\\
    c(b,t_{k+1},P_{-i})&=c(b,t_k,P_{-i})-r_k|N(b,t_k,P_{-i})|\\
    &=c(b,t_k,P)-r_k|N(b,t_k,P)|\\
    &=c(b,t_{k+1},P)\geq 0,
\end{align*}
which establishes \ref{neginvolvement:ainequality} and \ref{neginvolvement:bequality} for round $k+1$ (if $k$ is not the last round) and for $t_{k+1}=m/(n+1)$ (if $k$ is the last round) and also shows that the set of candidates with positive capacity is the same for both profiles at the starting time of round~$k+1$. Hence, at the starting time of round~$k+1$, each voter except $i$ is eating the same candidate in both profiles and $i$ is eating $a$ in~$P$, which establishes  \ref{neginvolvement:aN} and \ref{neginvolvement:bN} for round $k+1$.

Now observe that if we apply \ref{neginvolvement:ainequality} and \ref{neginvolvement:bequality} for $t_{k+1}=m/(n+1)$, we get that $c(b,m/(n+1),P_{-i})=c(b,m/(n+1),P)=0$ and $c(a,m/(n+1),P_{-i})\geq c(a,m/(n+1),P)=0$. However, the algorithm under $P_{-i}$ will not terminate until time $m/n$. Hence $c(a,m/(n+1),P_{-i})>0$, and $a$ is the unique candidate elected as required.
\end{proof}

\rrualinvolvement*
\begin{proof}
Let $P$ be a profile with $n$ voters and $P_{-i}$ the profile with $n-1$ voters obtained by removing voter $i$ from $P$. Choose $P$ and $P_{-i}$ such that they illustrate the failure of voter adaptability/positive involvement by~$\rho$.

Consider the profiles $P'_{-i}$ and $P'$ that result from $P_{-i}$ and $P$ by adding $m(n-1)$ unanimous losers. The veto cores of $P'_{-i}$ and $P'$ consist of all Pareto-optimal candidates of $P_{-i}$ and $P$ respectively. As such, the same candidates win in $\CR_\rho(P'_{-i})$ and $\CR_\rho(P')$ as in $\rho(P_{-i})$ and $\rho(P)$ and the same property is violated.
\end{proof}

\HHAWoodall*
\begin{proof}
We will first observe that at the moment a candidate is completely eaten, at least $\frac{n}{m}$ voters are eating this candidate. Recall that the algorithm lasts for precisely $\frac{m}{n}$ units of time, so if less than $\frac{n}{m}$ voters ate $c$ throughout the algorithm, then $c$ would have non-zero capacity at the end, which is impossible.

Fix the number of candidates $m\geq 2$. We will show the next result: the probability that exactly one candidate is eaten at the end of each round tends to 1 as $n\ria\infty$. Because the number of pairs of candidates does not exceed $m(m-1)/2$, it is sufficient to prove the next statement: the probability that candidates 1 and 2 are eaten simultaneously limits to 0 as $n\ria\infty$.

Consider the beginning of an arbitrary round. Let $c_i$ be the remaining capacity of candidate $i$, $n_i$ the number of voters who were eating $i$ before the beginning of the round, and $s_i$ the number of new voters who started eating $i$ this round. A necessary condition for $c_1$ and $c_2$ to be eaten simultaneously is:
\begin{equation*}
    \frac{n_1+s_1}{c_1}=\frac{n_2+s_2}{c_2}.
\end{equation*}
Let $r$ be the number of candidates with non-zero capacities at the beginning of the round, and $N=s_1+\dots+s_r$ the number of voters who started eating a new candidate this round. Since at least one candidate was eaten in the last round, we have seen above that $N\geq\frac{n}{m}$, and hence $N\ria\infty$ as $n\ria\infty$.

It is sufficient to show that for each fixed $r\geq 2$, as $N\ria\infty$
\begin{equation*}
    \sup\limits_{c_1,c_2\in\RR_{++},n_1,n_2\in\RR_{+}}{\PP\left[\frac{n_1+s_1}{c_1}=\frac{n_2+s_2}{c_2}\right]}\ria 0.
\end{equation*}
Let $b=\frac{c_1}{c_2}>0$, and $x^b_{1N},\ldots,x^b_{NN}$ be mutually independent and identically distributed variables such that $x^b_{kN}$ is equal to 1 and $-b$ with probabilities $1/r$ each, and 0 with probability $(r-2)/r$. Observe that $y_N^b=\sum x^b_{kN} = s_1-bs_2$. Letting $d=\frac{c_1}{c_2}n_2-n_1$, the above is equivalent to:
\begin{equation*}
    \sup\limits_{b\in\RR_{++},d\in\RR}{\PP[y^b_N=d]}\ria 0.
\end{equation*}
By introducing $z^b_N=(y^b_N-\EE[y^b_N])/\sqrt{Var[y^b_N]}$ and $c=(d-\EE[y^b_N])/\sqrt{Var[y^b_N]}$, we can rewrite the condition again as:
\begin{equation*}
    \sup\limits_{b\in\RR_{++},c\in\RR}{\PP[z^b_N=c]}\ria 0.
\end{equation*}

We will prove it by contradiction. Suppose there exist a fixed $\varepsilon>0$ and an infinitely increasing sequence of natural numbers, $N_j$, $j=1,2,\ldots$, such that for each $j$
\begin{equation*}
    \sup\limits_{b\in\RR_{++},c\in\RR}{\PP[z^b_{N_j}=c]}>\varepsilon.
\end{equation*}
Then we can fix a pair of sequences $b_N>0$ and $c_N$, $N=1,2,\ldots$, such that 
\begin{equation}\label{not_limits}
    \PP[z^{b_N}_{N}=c_N]>\varepsilon\quad\text{for each $N\in\{N_1,N_2,\ldots\}$}.
\end{equation}
Consider $z_N=z^{b_N}_N, y_N=y^{b_N}_N, x_{kN}=x^{b_N}_{kN}$. Note that
\begin{align*}
    z_N=\frac{y_N-\EE[y_N]}{\sqrt{Var[y_N]}}=\sum\limits_{k=1}^N{\frac{x_{kN}-\EE[x_{kN}]}{\sqrt{Var[x_{kN}]\cdot N}}},
\end{align*}
where
\begin{align*}
    \EE[x_{kN}]&=\frac{1-b_N}{r},\\
    Var[x_{kN}]=\EE[(x_{kN})^2]-(\EE[x_{kN}])^2&=\frac{1+(b_N)^2}{r}-\frac{(1-b_N)^2}{r^2}\\
    &=\frac{(r-1)(b_N)^2+2b_N+r-1}{r^2}.
\end{align*}
Let us show that
\begin{align*}
    \left|\frac{x_{kN}-\EE[x_{kN}]}{\sqrt{Var[x_{kN}]}}\right|&\leq\frac{\max\left\{\left|1-\frac{1-b_N}{r}\right|,\left|-b_N-\frac{1-b_N}{r}\right|\right\}}{\frac{1}{r}\sqrt{(r-1)(b_N)^2+2b_N+r-1}}\\
     &=\frac{\max\{r-1+b_N,(r-1)b_N+1\}}{\sqrt{(r-1)(b_N)^2+2b_N+r-1}}\leq\frac{r}{\sqrt{r-1}},
\end{align*}
where for $0<b_N\leq 1$ the latter inequality is straightforward, whereas for $1\leq b_N$ it follows from the fact that $r-1+b_N\leq(r-1)b_N+1$ and:
\begin{equation*}
    \frac{(r-1)b_N+1}{\sqrt{(r-1)(b_N)^2+2b_N+r-1}}\leq\frac{rb_N}{\sqrt{(r-1)(b_N)^2}}=\frac{r}{\sqrt{r-1}}.
\end{equation*}
For a fixed $s>2$,
\begin{align*}
    \sum\limits_{k=1}^N{\EE\left[\left|\frac{x_{kN}-\EE[x_{kN}]}{\sqrt{Var[x_{kN}]\cdot N}}\right|^s\right]}\leq
    N\left(\frac{r}{\sqrt{r-1}}\right)^s\frac{1}{\sqrt{N^s}}=\left(\frac{r}{\sqrt{r-1}}\right)^s\frac{1}{\sqrt{N^{s-2}}},
\end{align*}
which limits to 0 as $N\ria\infty$, and hence by Lyapunov's condition of the Central Limit Theorem \citep[p. 221-227]{BureParilina13bookEnglish}:%(Bure and Parilina, 2013, p. 221-227) %\citep[p.221-227]{BureParilina13book}
\begin{equation*}
    \sup\limits_{x\in\RR}{|\PP[z_N\leq x]-F(x)|}\ria 0.
\end{equation*}
Then:
\begin{align*}
    \PP[z_N=c_N]\leq\PP[z_N\leq c_N]-\PP\left[z_N\leq c_N-\frac{1}{N}\right]=\\
    \PP[z_N\leq c_N]-F(c_N)+F(c_N)-F\left(c_N-\frac{1}{N}\right)+F\left(c_N-\frac{1}{N}\right)-\PP\left[z_N\leq c_N-\frac{1}{N}\right]\leq\\
    \left|\PP[z_N\leq c_N]-F(c_N)\right|+\frac{1}{\sqrt{2\pi}}\int\limits_{c_N-\frac{1}{N}}^{c_N}{e^{-\frac{t^2}{2}}dt}+\left|\PP\left[z_N\leq c_N-\frac{1}{N}\right]-F\left(c_N-\frac{1}{N}\right)\right|\leq\\
    2\sup\limits_{x\in\RR}{|\PP[z_N\leq x]-F(x)|}+\frac{1}{N\sqrt{2\pi}},
\end{align*}
which limits to 0 as $N\ria\infty$ and hence contradicts to (\ref{not_limits}).
\end{proof}

\VbCsizentwo*
\begin{proof}
    Denote $a_m$ the average number of veto by consumption winners for 2 voters and $m$ candidates. Obviously, $a_1=1$ and $a_2=3/2$. The share of profiles with the same least-preferred candidate for both voters is $1/m$, and after eating this candidate the situation is reduced to $m-1$ candidates. The share of profiles with different least-preferred candidates for the two voters is $(m-1)/m$, and after eating these two candidates the situation is reduced to $m-2$ candidates. Hence, we have the next recurrent formula:
    \begin{align*}
        a_m=\frac{1}{m}a_{m-1}+\frac{m-1}{m}a_{m-2},
    \end{align*}
    which implies that:
    \begin{align*}
        a_m-a_{m-1}&=-\frac{m-1}{m}(a_{m-1}-a_{m-2})\\
        &=\left(-\frac{m-1}{m}\right)\left(-\frac{m-2}{m-1}\right)(a_{m-2}-a_{m-3})\\
        &=(-1)^2\frac{(m-2)}{m}(a_{m-2}-a_{m-3})\\
        &=(-1)^m\frac{2}{m}(a_2-a_1)\\
        &=\frac{(-1)^m}{m}.
    \end{align*}
    The proposition follows from the fact that:
    \begin{align*}
        a_m&=a_1+(a_2-a_1)+\ldots+(a_m-a_{m-1})\\
        &=1+\frac{1}{2}-\frac{1}{3}+\ldots+\frac{(-1)^m}{m}.
    \end{align*}
\end{proof}

\begin{restatable}{lemma}{GAlargem}\label{lemma:GA-large-m}
With the Impartial Culture assumption and a fixed number of voters $n\neq 2$, as $m\ria\infty$, the expected number of tied pairs of candidates under generalised antiplurality tends to~$0$. 
    
For $n=2$, the expected number of Pareto-optimal tied pairs of candidates under generalised antiplurality tends to~$0$.
\end{restatable}
\begin{proof} %The case of $n=1$ is trivial.

Fix $n\geq 3$. Consider a number of candidates~$m$ and a pair of candidates $a,b$. For each voter $i$, denote $a_i$ the position at which this voter ranks~$a$. Generalised antiplurality ranks $a$ and $b$ equally only if $\set{a_1,\ldots,a_n}=\set{b_1,\ldots,b_n}$. Since voter $i$ ranks $a$ at $a_i$, in order for $a$ and $b$ to be tied, voter $i$ must rank $b$ at one of the positions in $\set{a_1,\ldots,a_{i-1},a_{i+1},\ldots,a_n}$, which happens with probability no higher than $(n-1)/(m-1)$. The probability that every voter ranks $b$ at one of the positions in $\set{a_1,\ldots,a_n}$ is no higher than $(n-1)^n/(m-1)^n$. The total number of pairs of candidates is $m(m-1)/2$. Hence, the expected number of tied pairs of candidates is no higher than:
    $$\frac{m(m-1)}{2}\frac{(n-1)^n}{(m-1)^n}=\frac{m(n-1)^n}{2(m-1)}\frac{1}{(m-1)^{n-2}},$$
which tends to 0, as $m\ria\infty$.

Fix $n=2$. Consider $m\geq 4$ and a pair of candidates~$a,b$. They are tied under generalised antiplurality if and only if one voter ranks $a$ and $b$ at positions $x$ and $y$, whereas the second voter at positions $y$ and $x$, respectively, which happens with probability $1/m(m-1)$. Let us estimate $P_m$, the conditional probability that $a$ is Pareto-optimal provided that $a$ and $b$ are tied. Observe that $a$ and $b$ occupy each of $m(m-1)/2$ position pairs ($1\leq x<y\leq m$) with equal probability, so the probability of choosing a specific pair $(x,y)$ is $\frac{2}{m(m-1)}$. If $x=1$ then $a$ is definitely Pareto optimal, and there are $m-1$ pairs $(x,y)$ with $x=1$. Consider $x\geq 2$. Candidate $a$ is Pareto-optimal if and only if one voter ranks $a$ at position $x$ and candidates $c_1,\ldots,c_{x-1}$ above position $x$, whereas the second voter ranks $c_1,\ldots,c_{x-1}$ below position $y$, which happens with probability:
\begin{align*}
    \frac{(m-y)(m-y-1)\ldots(m-y-x+2)}{(m-2)(m-3)\ldots(m-x)}.
\end{align*}
Hence,
\begin{align*}
    P_m=\frac{2}{m(m-1)}\left(m-1+\sum_{x=2}^{\floor{m/2}}\sum_{y=x+1}^{m-x+1}\frac{(m-y)\ldots(m-y-x+2)}{(m-2)\ldots(m-x)}\right)
\end{align*}
and the probability that $a$ and $b$ are tied and Pareto-optimal is not higher than $P_m/m(m-1)$. Because the number of pairs of candidates is $m(m-1)/2$, the expected number of Pareto-optimal tied pairs of candidates is not higher than
\begin{align*}
    \frac{m(m-1)}{2}\frac{P_m}{m(m-1)}=\frac{P_m}{2}&=\frac{1}{m}+\frac{1}{m(m-1)}\sum_{x=2}^{\floor{m/2}}\sum_{y=x+1}^{m-x+1}\frac{(m-y)\ldots(m-y-x+2)}{(m-2)\ldots(m-x)}\\
    &\leq \frac{1}{m}+\frac{1}{(m-2)^2}\sum_{x=2}^{\floor{m/2}}\sum_{y=x+1}^{m-x+1}\left(\frac{m-y}{m-2}\right)^{x-1}\\
    &\leq \frac{1}{m}+\frac{1}{(m-2)^2}\sum_{x=2}^{\floor{m/2}}\sum_{y=x+1}^{m-x+1}\left(\frac{m-x-1}{m-2}\right)^{x-1}\\
     &\leq \frac{1}{m}+\frac{1}{(m-2)^2}\sum_{x=2}^{\floor{m/2}}(m-2x+1)\left(\frac{m-x-1}{m-2}\right)^{x-1}\\
     &\leq \frac{1}{m}+\frac{1}{(m-2)}\sum_{x=2}^{\floor{m/2}}\left(\frac{m-x-1}{m-2}\right)^x\\
     &\leq \frac{1}{m}+\frac{1}{(m-2)}\sum_{x=2}^{m-2}\left(1-\frac{x-1}{m-2}\right)^x\\
     &= \frac{1}{m}+\frac{1}{(m-2)}\sum_{z=1}^{m-3}\left(1-\frac{z}{m-2}\right)^{z+1}\\
     &= \frac{1}{k+2}+\frac{1}{k}\sum_{z=1}^{k-1}\left(1-\frac{z}{k}\right)^{z+1},
\end{align*}
where $k=m-2$. We will now show that the above tends to 0 as $k\ria\infty$. 

For each $\varepsilon>0$, we can choose an integer $p\geq 3$ such that $\varepsilon\geq 2/(p+1)$. Then for each $k\geq (p+1)(p-2)$, we have that:
\begin{align*}
    \frac{1}{k}\sum_{z=1}^{k-1}\left(1-\frac{z}{k}\right)^{z+1}&=\frac{1}{k}\sum_{z=1}^{p-2}\left(1-\frac{z}{k}\right)^{z+1}+\frac{1}{k}\sum_{z=p-1}^{k-1}\left(1-\frac{z}{k}\right)^{z+1}\\
    &\leq\frac{p-2}{k}+\frac{1}{k}\sum_{z=p-1}^{k-1}\left(1-\frac{z}{k}\right)^p\\
    &\leq\frac{1}{p+1}+\frac{1}{k}\sum_{z=1}^{k-1}\left(1-\frac{z}{k}\right)^p\\
    &=\frac{1}{p+1}+\frac{1}{k}\sum_{z=1}^{k-1}\left(\frac{z}{k}\right)^p\\
    &\leq\frac{1}{p+1}+\int_0^1t^pdt=\frac{2}{p+1}\leq \varepsilon.
\end{align*}
\end{proof}

\GAtiebreaking*
\begin{proof}
Fix $n\geq 3$. The number of winners under any voting rule with generalised antiplurality tie-breaking is no higher than one plus the number of tied pairs of candidates under generalised antiplurality. Hence, the expected number of winners under any voting rule with generalised antiplurality tie-breaking is no higher than one plus the expected number of tied pairs of candidates under generalised antiplurality, which tends to one plus zero by \autoref{lemma:GA-large-m}.

For $n=2$, the argument is identical.
\end{proof}

\HHAGAposresponse*
\begin{proof}
For a profile $P$, let $B$ be the set of winners under veto by consumption and $A\subseteq B$ the set of winners under veto by consumption with generalised antiplurality tie-breaking. Let a profile $P'$ be obtained from $P$ by having one voter raise a candidate $c\in A$ one position in his voting order. Because generalised antiplurality ranks $c$ no lower than any other candidate from $B$ in $P$, it ranks $c$ higher than any other candidate from $B$ in $P'$. Denote $B'$ the set of winners under veto by consumption in~$P'$. By \autoref{prop:fractionalveto} and the proof in \autoref{prop:AVmonotonic}, $c\in B'$ and $B'\subseteq B$. Hence, generalised antiplurality ranks $c$ higher than any other candidate from $B'$ in $P'$ and $c$ is the unique winner under veto by consumption with generalised antiplurality tie-breaking.  
\end{proof}

\section{Proof of \autoref{lem:tokenscore}}\label{app:tokens}

\autoref{lem:tokenscore} is morally similar to theorem 4-bis of \citet[page 151]{Moulin1983book}, and is proved using similar techniques.

We start with a technical lemma. \citet{Moulin1983book} contains a sketch of the proof, we reconstruct the full proof below.
\begin{lemma}[{\citealt[lemma 4, page 145]{Moulin1983book}}]\label{lem:marriage}
Let $X$ and $N$ be two sets, $|X|\geq |N|$. Suppose there exist sets $X_i\subseteq X$ for $i\in N$ with the following property:
\begin{equation}\label{eqn:Tlessunion}
    \forall T\subseteq N: |T|\leq \left|\bigcup_{i\in T} X_i\right|.
\end{equation}
There exists an $i'$ such that:
\begin{equation}\label{eqn:marriagelemma}
    \forall T\subseteq N\setminus\set{i'}: \left|\bigcup_{i\in T} X_i\right|=|T| \Rightarrow X_{i'}\cap\left(\bigcup_{i\in T} X_i\right)=\emptyset.
\end{equation}
\end{lemma}
\begin{proof}
Let $X(T)=\bigcup_{i\in T} X_i$, $\CT=\set{T : |T|=|X(T)|}$. Observe that \ref{eqn:marriagelemma} can be restated as:
\begin{equation}\label{eqn:marriagelemmamk2}
     \forall T\subseteq N\setminus\set{i'}: T\in\CT \Rightarrow X_{i'}\cap X(T)=\emptyset.
\end{equation}

Consider $T_1,T_2\in\CT$. We first establish two facts about the structure of $\CT$ (\ref{eqn:intersectionnotempty}, \ref{eqn:intersectionempty}):

\begin{equation}\label{eqn:intersectionnotempty}
    T_1\cap T_2\neq\emptyset\Rightarrow T_1\cap T_2\in\CT.
\end{equation}
Let $x=|T_1|=|X(T_1)|$, $y=|T_2|=|X(T_2)|$, $z=|T_1\cap T_2|$, $w=|X(T_1\cap T_2)|$. By \ref{eqn:Tlessunion}, $w\geq z$. For contradiction, suppose $w>z$. Let us consider the size of $X(T_1\cup T_2)$:
\begin{align*}
    |X(T_1\cup T_2)| &= \left|\bigcup_{i\in T_1\cup T_2} X_i\right|,\\
    &= \left|\bigcup_{i\in T_1} X_i \cup\bigcup_{i\in T_2}X_i \right|,\\
    &=|X(T_1)\cup X(T_2)|,\\
    &=|X(T_1)|+|X(T_2)|-|T_1\cup T_2|,\\
    &=x+y-w.
\end{align*}
However, observe that $|T_1\cup T_2|=x+y-z$, and by \ref{eqn:Tlessunion} we require that $x+y-z=|T_1\cup T_2|\leq |X(T_1\cup T_2)| =x+y-w$. Since $w>z$, this is impossible.

\begin{equation}\label{eqn:intersectionempty}
T_1\cap T_2=\emptyset\Rightarrow X(T_1)\cap X(T_2)=\emptyset.    
\end{equation}
If this were not the case, then $|T_1\cup T_2|>|X(T_1)\cup X(T_2)| =|X(T_1\cup T_2)|$, which contradicts \ref{eqn:Tlessunion}.

We will now prove the lemma by finding an $i'$ satisfying \ref{eqn:marriagelemmamk2}. Choose a $T\in\CT$ minimal with respect to inclusion. By \ref{eqn:intersectionnotempty}, every other $R\in\CT$ must either be a superset of $T$ or have an empty intersection with $T$. Choose any $i'\in T$. If $R$ is a superset of $T$, then $R\nsubseteq T\setminus\set{i'}$, so \ref{eqn:marriagelemmamk2} is satisfied vacuously. If $R\cap T=\emptyset$, then by \ref{eqn:intersectionempty}, $X(T)\cap X(R)=\emptyset$. Since $X_{i'}\subseteq X(T)$, \ref{eqn:marriagelemmamk2} is satisfied.
\end{proof}

\begin{definition}\label{def:trcore}
Let $\boldsymbol{\lambda}=(\lambda_1,\dots,\lambda_m)$ and $\boldsymbol{\mu}=(\mu_1,\dots,\mu_n)$ be vectors of non-negative integers, with the intuitive meaning that $\mu_i$ is the number of veto tokens voters $i$ can cast and $\lambda_i$ is how many clones candidate $c_i$ has. Let $\boldsymbol{\mu}(T)=\sum_{i\in T} \mu_i$, $\boldsymbol{\lambda}(B)=\sum_{c_i\in B} \lambda_i$.

Given an order $\sigma$ over all $\boldsymbol{\mu}(\CV)$ tokens, voting by $(\boldsymbol{\lambda},\boldsymbol{\mu})$-tokens proceeds by letting every voter cast a veto token in the order specified by $\sigma$. The winners are all candidates with a clone remaining at the end. 

Let $\text{Better}(T,P,c)$ be the set of all candidates that every voter in $T$ prefers to $c$. The $(\boldsymbol{\lambda},\boldsymbol{\mu})$-core of profile $P$ is all candidates $c$ for which:

$$\forall T\subseteq\CV: \boldsymbol{\lambda}(\text{Better}(T,P,c))+\boldsymbol{\mu}(T)<\boldsymbol{\lambda}(\CC).$$
\end{definition}

The difference between voting by $(\boldsymbol{\lambda},\boldsymbol{\mu})$-tokens and voting by veto tokens as per \autoref{def:vetotokens},  is that in voting by veto tokens every voter has the same $r$ number of veto tokens and every candidate the same $t$ clones, $r$ and $t$ chosen to satisfy \autoref{lem:moulin}, and thus after the vetoes are cast there are exactly $\gcd(m,n)$ clones of candidates remaining; voting by $(\boldsymbol{\lambda},\boldsymbol{\mu})$-tokens allows an arbitrary distribution of veto power, and we make no restriction as to how many clones are left after all voters have cast their veto tokens.

The following is very similar to theorem 4 in \citet[page 143]{Moulin1983book}:

\begin{lemma}\label{lem:11tokens}
Candidate $c$ is in the $(\boldsymbol{1},\boldsymbol{1})$-core of $P$ if and only if there exists an order of $(\boldsymbol{1},\boldsymbol{1})$-tokens that elects $c$.
\end{lemma}
\begin{proof}
Suppose $c$ is not in the $(\boldsymbol{1},\boldsymbol{1})$-core. That means there exists a $T\subseteq\CV$ such that:
$$\boldsymbol{1}(\text{Better}(T,P,c))+\boldsymbol{1}(T)\geq\boldsymbol{1}(\CC),$$
or equivalently:
$$|T|\geq|\CC\setminus\text{Better}(T,P,c)|.$$
A voter of $T$ will veto a candidate in $\CC\setminus\text{Better}(T,P,c)$ as long as $c$ remains. From the inequality above, it thus follows that by the time all voters in $T$ have cast their vetoes $c$ must be eliminated.

Suppose $c$ is in the $(\boldsymbol{1},\boldsymbol{1})$-core. Observe that if the $(\boldsymbol{1},\boldsymbol{1})$-core is non-empty, it must be that $|\CC|>|\CV|$.

Let $\text{Worse}(T,P,c)=\set{a\in\CC\mid \exists i\in T: a\prec_i c}$ be the set of candidates $a$ such that at least one voter in $T$ prefers $c$ to $a$. Observe that $\CC=\text{Better}(T,P,c)\uplus\set{c}\uplus\text{Worse}(T,P,c)$.

Choose a $c$ in the $(\boldsymbol{1},\boldsymbol{1})$-core. From \autoref{def:trcore}:
\begin{align}
    \forall T\subseteq \CV:\quad&|\text{Better}(T,P,c)|+|T|<|\CC|,\nonumber\\
    \forall T\subseteq \CV:\quad& |T|\leq |\text{Worse}(T,P,c)|.\label{eqn:Tworse}
\end{align}
That is, no coalition has enough members to eliminate every candidate worse than $c$ and $c$ itself.

Let $W_i=\text{Worse}(\set{i},P,c)$ be the set of candidates $i$ considers to be worse than $c$. Using \ref{eqn:Tworse} and the fact that $\bigcup_{i\in T} W_i=\text{Worse}(T,P,c)$, we have that:
$$\forall T\subseteq\CV: |T|\leq \left|\bigcup_{i\in T}W_i\right|.$$
Since $|\CV|\leq|\CC\setminus\set{c}|$, all the conditions of \autoref{lem:marriage} are satisfied and there exists a voter $i'$ such that:
\begin{equation}\label{eqn:mlemma}
    \forall T\subseteq \CV\setminus\set{i'}: \left|\bigcup_{i\in T} W_i\right|=|T| \Rightarrow W_{i'}\cap\left(\bigcup_{i\in T} W_i\right)=\emptyset.
\end{equation}
Applying \ref{eqn:Tworse} to the singleton $T=\set{i'}$, it follows that $W_{i'}$ is non-empty. Let $w$ be $i'$'s worst choice. Clearly, $w\in W_{i'}$. We will show that in profile $P'$ obtained by removing $i'$ and $w$, $c$ is still in the $(\boldsymbol{1},\boldsymbol{1})$-core. This will allow us to inductively construct an order that elects $c$: at each stage, choose an $i'$ via the lemma, and have that candidate veto his last choice, who is not $c$.

Suppose otherwise: that there exist a coalition $T\subseteq\CV\setminus\set{i'}$ for whom:
\begin{align}
    |T|&> |\text{Worse}(T,P',c)|,\nonumber\\
    |T|&> |\CC'| - |\text{Better}(T,P',c)|-1,\nonumber\\
    |T|&\geq |\CC| - |\text{Better}(T,P',c)| -1.\label{eqn:TCb1}
\end{align}

Observe $\text{Better}(T,P',c)\cup\set{c}$ and $\text{Worse}(T,P,c)$ are disjoint. It follows that $|\text{Worse}(T,P,c)|+|\text{Better}(T,P',c)|+1\leq |\CC|$. We combine this with \ref{eqn:TCb1} to obtain:
$$|T|\geq  |\CC| - |\text{Better}(T,P',c)| -1\geq|\text{Worse}(T,P,c)|.$$
Since $c$ is in the $(\boldsymbol{1},\boldsymbol{1})$-core of $P$, it also follows that $|\text{Worse}(T,P,c)|\geq |T|$, so the two cardinalities are equal. We can thus invoke \ref{eqn:mlemma} to show that $W_{i'}\cap\left(\bigcup_{i\in T} W_i\right)=\emptyset$.

Since $w\in W_{i'}$, and $w\neq c$, it follows that $w\in\text{Better}(T,P,c)$. I.e., all voters in $T$ prefer $w$ to $c$. However, $w\notin \text{Better}(T,P',c)$, because $w$ is absent in $P'$. Thus $|\text{Better}(T,P,c)|\geq |\text{Better}(T,P',c)|+1\geq |\CC|-|T|$. This is a contradiction, because $\text{Better}(T,P,c),\text{Worse}(T,P,c),\set{c}$ are supposed to partition the candidates, but now the sets are of sizes at least $|\CC|-|T|,|T|,1$ respectively. 
\end{proof}

The following is very similar to theorem 4-bis in \citet[page 151]{Moulin1983book}:

\begin{corollary}\label{cor:trcore11core}
Let $P_{\boldsymbol{\lambda},\boldsymbol{\mu}}$ be the profile obtained from $P$ by replacing every candidate $c_i$ with $\lambda_i$ clones of $c_i$, and every voter $i$ with $\mu_i$ clones of $i$. The following is true:
\begin{enumerate}
\item A candidate $c$ can be elected by $(\boldsymbol{\lambda},\boldsymbol{\mu})$-tokens in $P$ if and only if a clone of $c$ can be elected by $(\boldsymbol{1},\boldsymbol{1})$-tokens in $P_{\boldsymbol{\lambda},\boldsymbol{\mu}}$.
    \item A candidate $c$ is in the $(\boldsymbol{\lambda},\boldsymbol{\mu})$-core of $P$ if and only if a clone of $c$ is in the $(\boldsymbol{1},\boldsymbol{1})$-core of $P_{\boldsymbol{\lambda},\boldsymbol{\mu}}$.
    \item A candidate $c$ is in the $(\boldsymbol{\lambda},\boldsymbol{\mu})$-core of $P$ if and only if $c$ can be elected by $(\boldsymbol{\lambda},\boldsymbol{\mu})$-tokens in $P$.
\end{enumerate}
\end{corollary}
\begin{proof}
1 is obvious.

For 2, suppose $c$ is not in the $(\boldsymbol{\lambda},\boldsymbol{\mu})$-core of $P$. It follows there exists a $T$ for which:
$$\boldsymbol{\mu}(T)\geq \boldsymbol{\lambda}(\CC)- \boldsymbol{\lambda}(\text{Better}(T,P,c)).$$
Consider the coalition $T'$ consisting of the clones of $T$ in $P_{\boldsymbol{\lambda},\boldsymbol{\mu}}$. This coalition has $\boldsymbol{\mu}(T)$ members. $\text{Better}(T',P_{\boldsymbol{\lambda},\boldsymbol{\mu}},c')$ contains at least $\boldsymbol{\lambda}(\text{Better}(T,P,c))$ candidates for every clone clone $c'$ of $c$, and there are $\boldsymbol{\lambda}(\CC)$ candidates in $P_{\boldsymbol{\lambda},\boldsymbol{\mu}}$ in total. Thus:
$$|T'|=\boldsymbol{\mu}(T)\geq \boldsymbol{\lambda}(\CC)- \boldsymbol{\lambda}(\text{Better}(T,P,c)) = |\CC'|-\text{Better}(T',P_{\boldsymbol{\lambda},\boldsymbol{\mu}},c').$$
Any clone $c'$ of $c$ is not in the $(\boldsymbol{1},\boldsymbol{1})$-core of $P_{\boldsymbol{\lambda},\boldsymbol{\mu}}$.

Suppose every clone of $c$ is not in the $(\boldsymbol{1},\boldsymbol{1})$-core of $P_{\boldsymbol{\lambda},\boldsymbol{\mu}}$. Without loss of generality, suppose all the voters rank the clones $c_1\succ\dots\succ c_{\lambda_c}$, and consider a coalition $T'$ for which. 
$$|T'|\geq |\CC'|- |\text{Better}(T',P_{\boldsymbol{\lambda},\boldsymbol{\mu}},c_1)|.$$
Of all such $T'$, pick one that is clone-maximal, i.e.\ if it includes a clone of $i$, it includes all clones of $i$. Observe that $|T'|=\sum_{i\in T} \mu_i i=\boldsymbol{\mu}(T)$ for some $T$ in $P$. Since $c_1$ is better than all the clones of $c$, $|\text{Better}(T',P_{\boldsymbol{\lambda},\boldsymbol{\mu}},c_1)|=\boldsymbol{\lambda}(\text{Better}(T,P,c))$. It follows:

$$\boldsymbol{\mu}(T)=|T'|\geq |\CC'|-|\text{Better}(T',P',c_1)|=\boldsymbol{\lambda}(\CC)- \boldsymbol{\lambda}(\text{Better}(T,P,c)).$$
Thus $c$ is not in the $(\boldsymbol{\lambda},\boldsymbol{\mu})$-core of $P$.

For 3, combine 1 and 2 with \autoref{lem:11tokens}.
\end{proof}

Now all we need to do is to link \autoref{def:trcore} back to the proportional veto core of \autoref{def:vetocore}.

\begin{lemma}\label{lem:coretrcore}
Let $r,t$ be as in \autoref{lem:moulin}. A candidate $c$ is in the $(\boldsymbol{t},\boldsymbol{r})$-core of $P$ if and only if $c$ is in the proportional veto core.
\end{lemma}
\begin{proof}
Suppose $c$ is not in the proportional veto core. That is, there exists a coalition $T$ for whom $m-|\text{Better}(T,P,c))|\leq v(T)$. It follows:
\begin{align*}
 m-|\text{Better}(T,P,c)|&\leq \ceil{m\frac{|T|}{n}}-1,\\ 
 m-|\text{Better}(T,P,c)|&\leq \floor{\frac{r|T|}{t}},\\
 tm-t|\text{Better}(T,P,c)|&\leq t\floor{\frac{r|T|}{t}}\leq r|T|,\\
 tm&\leq r|T|+t|\text{Better}(T,P,c)|.
\end{align*}
So $c$ is not in the $(\boldsymbol{t},\boldsymbol{r})$-core either.

Suppose $c$ is in the proportional veto core. For every coalition $T$:
\begin{align}
     m-|\text{Better}(T,P,c)|&> \ceil{m\frac{|T|}{n}}-1,\nonumber\\ 
 m-|\text{Better}(T,P,c)|&> \floor{\frac{r|T|}{t}},\nonumber\\
 tm-t|\text{Better}(T,P,c)|&> r|T|-t\epsilon,\quad \epsilon\in\set{1/t,\dots,t-1/t},\nonumber\\
 tm+x&> r|T|+t|\text{Better}(T,P,c)|,\quad x\geq 1.\label{eqn:tmplusx}
\end{align}
Now suppose, for contradiction, that $c$ is nevertheless not in the $(\boldsymbol{t},\boldsymbol{r})$-core, and there exists a $T'$ for whom:
$$tm\leq r|T'|+t|\text{Better}(T',P,c)|.$$
Combining this with \ref{eqn:tmplusx}, it follows that in fact:
$$tm= r|T'|+t|\text{Better}(T',P,c)|.$$
We deduce:
\begin{align}
    tm-t|\text{Better}(T',P,c)|=r|T'|,\label{eqn:tmminustequal}\\
    \frac{r|T'|}{t}\text{ is an integer (divide both sides by $t$)}.\nonumber
\end{align}
However, we have assumed that:
\begin{align*}
     m-|\text{Better}(T',P,c)|&> \ceil{m\frac{|T'|}{n}}-1,\\ 
 m-|\text{Better}(T,P,c)|&> \floor{\frac{r|T'|}{t}}=\frac{r|T'|}{t}.
\end{align*}
Which contradicts \ref{eqn:tmminustequal}.
\end{proof}

This gives us our result.

\tokenscore*
\begin{proof}
Suppose $c$ is in the proportional veto core of $P$. By \autoref{lem:coretrcore}, $c$ s in the $(\boldsymbol{t},\boldsymbol{r})$-core of $P$. By \autoref{cor:trcore11core}, there exists an order of $(\boldsymbol{t},\boldsymbol{r})$-tokens that elects $c$ in $P$. $(\boldsymbol{t},\boldsymbol{r})$-tokens and voting by veto tokens are the same rule if we choose the functions $r(m,n),t(m,n)$ to return the values $r,t$, and the function $L(r(m,n),n)$ to return the same order as used by $(\boldsymbol{t},\boldsymbol{r})$-tokens.
\end{proof}

\section{Proof of \autoref{prop:HHAmonotonic}}\label{app:monotonic}

\begin{lemma}\label{lem:eatingtimes}
During the run of veto by consumption, a candidate can only be eaten at a time of the form $\frac{z}{n!^m}$, where $z$ is an integer.
\end{lemma}
\begin{proof}
We will show that the capacity of each candidate~$j\in\{1,\ldots,m\}$ at the start of each round~$k\in\{0,\ldots,K\}$ is of the form $\frac{x^k(j)}{n!^k}$ for some integer~$x^k(j)$. We proceed by induction. In the base case, the capacity of each candidate is $1/n!^0$ at the start of round~$0$. 

Suppose the inductive hypothesis holds for some round~$k$. That is, the capacity of each candidate~$j$ at the start of this round is $c(j)=\frac{x(j)}{n!^k}$ for some integer~$x(j)$, and $n(j)$ is the number of voters eating~$j$. If some candidate $i$ gets eaten, then the round will last for $r=\frac{c(i)}{n(i)}=\frac{x(i)}{n!^{k}n(i)}$ units of time. At the end of the round, the capacity of each candidate~$j$ is: \begin{equation*}
c(j)-r n(j)=\frac{x(j)}{n!^k}-\frac{x(i) n(j)}{n!^{k}n(i)}=\frac{x(j) n!-x(i) n(j) n!/n(i)}{n!^{k+1}},
\end{equation*}
where the numerator is an integer because $n(i)$ is a divisor of~$n!$.

The lemma will follow since, if a round~$k$ begins at time $\frac{z}{n!^m}$ for some integer~$z$, and some candidate $i$ gets eaten, then the round will last for $r=\frac{c(i)}{n(i)}=\frac{x(i)}{n!^{k}n(i)}$ units of time for some integer~$x(i)$, where $n(i)$ is the number of voters eating~$i$. It follows that the round ends at time
\begin{equation*}
    \frac{z}{n!^m}+\frac{x(i)}{n!^{k}n(i)}=\frac{s}{n!^m}
\end{equation*}
for some integer $s$ because $n(i)$ is a divisor of~$n!$.

\iffalse
We will show that the capacity of $c_i$ at the start of the $j$th round is of the form $\frac{k_i}{n!^j}$ for some integer $k_i$. The lemma will follow since, if the round begins at time $\frac{q_j}{n!^m}$, and some $c_w$ gets eaten, then the round will last for $\frac{k_w}{n!^jx_w}$ units of time, where $x_w$ is the number of voters eating $x_w$. Since $x_w|n!$, it follows that $\frac{q_j}{n!^m}+\frac{k_w}{n!^jx_w}=\frac{q_jn!^jx_w+k_wn!^m}{n!^mn!^jx_w}=\frac{q_{j+1}}{n!^m}$ for some integer $q_{j+1}$.

We proceed by induction. In the base case, the capacity of each candidate is $1/n!^0$ at the start of round 0. Suppose the capacity of $c_i$ at the start of the $j$th round is $\frac{k_i}{n!^j}$. As we have argued, the round lasts for $\frac{k_w}{n!^{j}x_w}$ units of time, during which $c_i$ is eaten by $x_i$ voters. At the end of the round, $c_i$'s capacity is: 
\begin{align*}
  \frac{k_i}{n!^j}-\frac{k_wx_i}{n!^{j}x_w}&=\frac{k_in!^{j}x_w-k_wx_in!^j}{n!^jn!^{j}x_w},\\
  &=\frac{k_in!x_w-k_wx_in!}{n!^{j+1}x_w},\\ 
  &=\frac{k_in!-k_wx_in!/x_w}{n!^{j+1}}.
\end{align*}
The numerator is an integer because $x_w|n!$.
\fi
\end{proof}

\begin{definition}%SUGGESTED BY ALEKSEI
Create $n!^m$ clones of every candidate. Repeating an order over the $n$ voters, % $\frac{mn!^m}{n}-1$ rounds, 
the voters successively veto a clone of their least preferred candidate until $n$ clones are left. \emph{Fractional veto} selects all candidates who have at least one clone remaining.
\end{definition}

\begin{proposition}\label{prop:fractionalveto}%SUGGESTED BY ALEKSEI We prove even more: 1) an order over the n voters can be different in different rounds; 2)  for each candidate, the number of his clones left in FV (among n clones left in the end) equals the number of voters eating him in the last round of veto by consumption.
For each order over the voters, fractional veto selects the same candidates as veto by consumption does.
\end{proposition}
\begin{proof}
We will show by induction that for every integer $z$ and candidate~$i$, this candidate has $x^z(i)$ clones left after $z$ rounds under fractional veto if and only if he has $\frac{x^z(i)}{n!^m}$ capacity left at time $\frac{z}{n!^m}$ under veto by consumption. In the base case, each candidate has $n!^m$ clones left after $0$ rounds under fractional veto and $\frac{n!^m}{n!^m}$ capacity left at time $\frac{0}{n!^m}$ under veto by consumption. 

Suppose the inductive hypothesis holds for some integer~$z$. That is, every candidate $i$ has $x^z(i)$ clones left after $z$ rounds under fractional veto, and $c(i)=\frac{x^z(i)}{n!^m}$ capacity left at time $\frac{z}{n!^m}$ under veto by consumption. Every candidate with $0$ clones left after $z$ rounds under fractional veto has also $0$ capacity left at time $\frac{z}{n!^m}$ under veto by consumption and is fully eliminated in both procedures, and hence will also have $0$ clones left after $z+1$ rounds under fractional veto and $0$ capacity left at time $\frac{z+1}{n!^m}$ under veto by consumption. 

Denote $X=\{i:x(i)>0\}$ the set of candidates who have at least one clone remaining after $z$ rounds under fractional veto; by the inductive hypothesis, this is also the of candidates with positive capacity at time $\frac{z}{n!^m}$ under veto by consumption. For each candidate $i\in X$, let $n(i)$ be the number of voters
who consider $i$ as their least preferred candidate in $X$, which means that they are eating $i$ under veto by consumption. By \autoref{lem:eatingtimes}, no candidate gets eaten between time $\frac{z}{n!^m}$ and time $\frac{z+1}{n!^m}$, so the same $n(i)$ voters eat $i$ over this time period, and $i$'s capacity reduces from $\frac{x(i)}{n!^m}$ to $\frac{x(i)-n(i)}{n!^m}\geq 0$. The latter inequality implies that for each candidate $i\in X$, the number of clones $x(i)$ is not less than the number of $n(i)$ voters willing to cast a veto against $i$ at the beginning of the $(z+1)$th round. As a result, for each order over the $n$ voters, exactly $n(i)$ voters cast a veto against $i$, leaving him with $x(i)-n(i)$ clones after $z+1$ rounds.
\end{proof}

\begin{definition}%SUGGESTED BY ALEKSEI
Create $S(j)$ clones of every candidate~$j$ and endow every voter $i$ with $T(i)$ veto tokens. Using an order over the veto tokens $L=(v_1,\dots,v_Z)$, the voters successively veto a clone of their least preferred candidate. \emph{Alternating veto} selects all candidates who have at least one clone remaining.
\end{definition}

Fractional veto is an alternating veto, where $S(j)=n!^m$ for every candidate~$j$, $T(i)=\frac{mn!^m}{n}-1$ for every voter $i$, and $L$ consists of repeating an order over the $n$ voters $\frac{mn!^m}{n}-1$ times.

Voting by veto tokens is an alternating veto, where $rn=tm-\gcd(m,n)$, $t>\gcd(m,n)n$, $S(j)=t$ for every candidate~$j$, $T(i)=r$ for every voter $i$, and $L$ is an order over the $rn$ veto tokens of the voters.

\begin{proposition}\label{prop:AVmonotonic}%SUGGESTED BY ALEKSEI Here we show even more: that our deterministic rule is strongly monotonic, and its corresponding probabilistic generalization, where probabilities are proportional to the number of clones left.
Alternating veto is monotonic.
\end{proposition}
\begin{proof}
Let $v_1,\ldots,v_Z$ be a voting order and $S(i)$ be an initial number of clones for each candidate~$i$. Call the original profile $P$. Let 1 be the winner under $P$, and let voter~V rank 1 directly below 2. We shall show that 1 remains the winner if $V$ switches 1 and 2.

Without loss of generality, extend the preferences in $P$ to preferences over clones of candidates by assuming that every voter ranks $i_1\succ\ldots\succ i_{S(i)}$ the clones of each candidate~$i$. call the resulting profile $P_0$. Let $S=S(1)$ and construct a sequence of profiles $P_0,\dots,P_S$, where $P_s$ is obtained from $P_0$ by shifting clones $1_1,\dots,1_s$ above the clones of 2 in voter~V's preference order, keeping other voters' preferences unchanged. This the preferences of voter~V look as follows:
\begin{align*}
    \ldots&\succ 2_1\succ\ldots\succ 2_{S(2)}\succ 1_1\succ\ldots\succ 1_{S}\succ\ldots&\text{for}\quad&P_0,\\
    \ldots\succ 1_1\succ\ldots\succ 1_s&\succ 2_1\succ\ldots\succ 2_{S(2)}\succ 1_{s+1}\succ\ldots\succ 1_{S}\succ\ldots&\text{for}\quad&P_s,\\
    \ldots\succ 1_1\succ\ldots\succ 1_{s+1}&\succ 2_1\succ\ldots\succ 2_{S(2)}\succ 1_{s+2}\succ\ldots\succ 1_{S}\succ\ldots&\text{for}\quad&P_{s+1},\\
    \ldots\succ 1_1\succ\ldots\succ 1_{S}&\succ 2_1\succ\ldots\succ 2_{S(2)}\succ\ldots&\text{for}\quad&P_{S}.
\end{align*}
Note that only $P_0$ and $P_S$ are true profiles in the sense that they are extensions of voters' preferences over candidates to preferences over clones, but the alterating veto procedure is well defined for all of them.

For each $z\in\{0,\ldots,Z\}$, $s\in\{0,\ldots,S\}$ and $i\in\{1,\ldots,m\}$, let $x^{z,s}(i)$ denote the number of clones of candidate~$i$ left after casting the first $z$ vetoes under profile~$P_s$. A single veto token $v_{z+1}$ vetoes exactly one clone and thus:
\begin{equation}\label{xzs_inequality}
    S(i)=x^{0,s}(i)\geq x^{1,s}(i)\geq\ldots\geq x^{Z,s}(i) \quad\text{for each}\quad s,i,
\end{equation}
\begin{equation}\label{deleting_clone}
    x^{z+1,s}(i)=\begin{cases} 
      x^{z,s}(i)-1& \text{if $v_{z+1}$ vetoes a clone of $i$}, \\
      x^{z,s}(i)& \text{otherwise}. 
   \end{cases}
\end{equation}

Since candidate 1 wins under profile~$P_0$ it follows that:
\begin{equation}\label{candidate1wins}
   x^{Z,0}(1)>0. 
\end{equation}
We want to show that $x^{Z,S}(1)\geq x^{Z,0}(1)$ and $x^{Z,S}(i)\leq x^{Z,0}(i)$ for each other candidate~$i$. This will establish that 1 remains the winner under $P_S$, and hence alternating veto is monotonic.

It will follow from the following claim:

\medskip
\noindent\textbf{Claim}. For each pair of $z\in\{0,\ldots,Z\}$ and $s\in\{0,\ldots,S-1\}$, we have that
\begin{equation}\label{claim_condition1}
    x^{Z,s}(1)>0,
\end{equation}
and one of the next two conditions is satisfied:
\begin{equation}\label{claim_condition2}
    x^{z,s+1}(i)=x^{z,s}(i) \quad \text{for each}\quad i,
\end{equation}
or
\begin{align}\label{claim_condition3}
   x^{z,s+1}(1)&=x^{z,s}(1)+1,&&\\
   \label{claim_condition4}
   x^{z,s+1}(j)&=x^{z,s}(j)-1&\text{for some}\quad&j\neq 1,\\
   \label{claim_condition5}
   x^{z,s+1}(i)&=x^{z,s}(i)&\text{for each}\quad&i\notin\{1,j\}.
\end{align}

\medskip
We prove the claim by induction. In the base case ($z=0$, $s=0$), (\ref{claim_condition1}) follows from~(\ref{candidate1wins}), and (\ref{claim_condition2}) is true because no vetoes have been cast at the beginning of the procedure.

Suppose the claim is true for some pair $(z,s)$ with $z<Z$. We will prove that the claim is true for $(z+1,s)$. Condition~(\ref{claim_condition1}) depends only on $s$, so it is true for $(z+1,s)$ if it is true for $(z,s)$.

By~(\ref{xzs_inequality}) and~(\ref{claim_condition1}), $x^{z,s}(1)>0$. By~(\ref{claim_condition2}) and (\ref{claim_condition3}), $x^{z,s+1}(1)>0$. We proceed by cases based on who casts the next veto.

\medskip
\setlength{\leftskip}{0 em}\noindent\textbf{Case 1}: $v_{z+1}$ is \textit{not} Voter~V and hence has \textit{identical} preferences in $P_s$ and $P_{s+1}$.

\medskip
\setlength{\leftskip}{1 em}\noindent\textbf{Case 1.1}: $x^{z,s+1}=x^{z,s}$. Obviously, $v_{z+1}$ vetoes the same clone in both situations. Hence, $x^{z+1,s+1}=x^{z+1,s}$ and (\ref{claim_condition2}) is true for $(z+1,s)$.

\medskip
\setlength{\leftskip}{1 em}\noindent\textbf{Case 1.2}: (\ref{claim_condition3})-(\ref{claim_condition5}) are true for $(z,s)$. By~(\ref{claim_condition4}), $x^{z,s}(j)>0$. %We have the next three possibilities.

\medskip
\setlength{\leftskip}{2 em}\noindent\textbf{Case 1.2.1}: $v_{z+1}$ vetoes a clone of $j$ in $x^{z,s}$. By~(\ref{deleting_clone}), $x^{z+1,s}(1)=x^{z,s}(1)$, $x^{z+1,s}(j)=x^{z,s}(j)-1$, and $x^{z+1,s}(i)=x^{z,s}(i)$ for each $i\notin\{1,j\}$. Thus, if $v_{z+1}$ vetoes a clone of~$1$ in $x^{z,s+1}$, then $x^{z+1,s+1}=x^{z+1,s}$ and (\ref{claim_condition2}) is true for $(z+1,s)$; otherwise, (\ref{claim_condition3})-(\ref{claim_condition5}) are true for $(z+1,s)$, but maybe for a new~$j$.

\medskip
\setlength{\leftskip}{2 em}\noindent\textbf{Case 1.2.2}: $v_{z+1}$ vetoes a clone of some $l\neq j$ in $x^{z,s}$.  By~(\ref{claim_condition5}), for each $i\neq j$, we have that $x^{z,s+1}(i)>0$ if and only if $x^{z,s}(i)>0$. Thus, $v_{z+1}$ vetoes a clone of~$l$ in $x^{z,s+1}$, and (\ref{claim_condition3})-(\ref{claim_condition5}) are true for $(z+1,s)$.

\medskip
\setlength{\leftskip}{0 em}\noindent\textbf{Case 2}: $v_{z+1}$ is Voter~V and hence has \textit{different} preferences in $P_s$ and $P_{s+1}$. Denote~$A$ the set of all candidates this voter prefers less than candidates 1 and~2.

\medskip
\setlength{\leftskip}{1 em}\noindent\textbf{Case 2.1}: $x^{z,s+1}=x^{z,s}$. %We have the next three possibilities.

\medskip
\setlength{\leftskip}{2 em}\noindent\textbf{Case 2.1.1}: $v_{z+1}$ vetoes a clone of $1$ in $x^{z,s}$. By~(\ref{deleting_clone}), $x^{z+1,s}(1)=x^{z,s}(1)-1$, and $x^{z+1,s}(i)=x^{z,s}(i)$ for each $i\neq 1$. Thus, if $v_{z+1}$ vetoes a clone of~$1$ in $x^{z,s+1}$, then $x^{z+1,s+1}=x^{z+1,s}$ and (\ref{claim_condition2}) is true for $(z+1,s)$; otherwise, (\ref{claim_condition3})-(\ref{claim_condition5}) are true for $(z+1,s)$.

\medskip
\setlength{\leftskip}{2 em}\noindent\textbf{Case 2.1.2}: $v_{z+1}$ vetoes a clone of~$2$ in $x^{z,s}$. Hence, $x^{z,s+1}(1)=x^{z,s}(1)\leq s$, $x^{z,s+1}(2)=x^{z,s}(2)>0$ and $x^{z,s+1}(i)=x^{z,s}(i)=0$ for each $i\in A$. Thus, $v_{z+1}$ vetoes a clone of~$2$ in $x^{z,s+1}$, $x^{z+1,s+1}=x^{z+1,s}$ and (\ref{claim_condition2}) is true for $(z+1,s)$.

\medskip
\setlength{\leftskip}{2 em}\noindent\textbf{Case 2.1.3}: $v_{z+1}$ vetoes a clone of some $l\in A$ in $x^{z,s}$. Hence, $x^{z,s+1}(l)=x^{z,s}(l)>0$ and $x^{z,s+1}(i)=x^{z,s}(i)=0$ for each $i$ such that $v_{z+1}$ prefers $l$ to $i$. Thus, $v_{z+1}$ vetoes a clone of~$l$ in $x^{z,s+1}$, $x^{z+1,s+1}=x^{z+1,s}$ and (\ref{claim_condition2}) is true for $(z+1,s)$.

\medskip
\setlength{\leftskip}{1 em}\noindent\textbf{Case 2.2}: (\ref{claim_condition3})-(\ref{claim_condition5}) are true for $(z,s)$. By~(\ref{claim_condition4}), $x^{z,s}(j)>0$. %We have the next three possibilities.

\medskip
\setlength{\leftskip}{2 em}\noindent\textbf{Case 2.2.1}: $v_{z+1}$ vetoes a clone of $j$ in $x^{z,s}$. See the argument in Case~1.2.1.

\medskip
\setlength{\leftskip}{2 em}\noindent\textbf{Case 2.2.2}: $v_{z+1}$ vetoes $1_k$ in $x^{z,s}$. If $k\leq s$, then $x^{z,s}(i)=0$ for each $i\in A\cup 2$, and by~(\ref{claim_condition4}) and (\ref{claim_condition5}),
$x^{z,s+1}(i)=0$ for each $i\in A\cup 2$. If $k\geq s+1$, then $x^{z,s}(1)\geq s+1$, $x^{z,s}(i)=0$ for each $i\in A$, and by~(\ref{claim_condition3}), $x^{z,s+1}(1)\geq s+2$, and  by~(\ref{claim_condition4}) and (\ref{claim_condition5}), $x^{z,s+1}(i)=0$ for each $i\in A$. Thus, $v_{z+1}$ vetoes a clone of~$1$ in $x^{z,s+1}$, and (\ref{claim_condition3})-(\ref{claim_condition5}) are true for $(z+1,s)$.

\medskip
\setlength{\leftskip}{2 em}\noindent\textbf{Case 2.2.3}: $v_{z+1}$ vetoes a clone of some $l\notin\{1,j\}$ in $x^{z,s}$. By~(\ref{claim_condition5}), $x^{z,s+1}(l)=x^{z,s}(l)>0$. If $l=2$, then $x^{z,s}(1)\leq s$ and $x^{z,s}(i)=0$ for each $i\in A$. By~(\ref{claim_condition3}), $x^{z,s+1}(1)\leq s+1$, and  by~(\ref{claim_condition4}) and (\ref{claim_condition5}), $x^{z,s+1}(i)=0$ for each $i\in A$. If $l\in A$, then $x^{z,s}(i)=0$ for each $i\in A$ such that $v_{z+1}$ prefers $l$ to $i$. By~(\ref{claim_condition4}) and (\ref{claim_condition5}), $x^{z,s+1}(i)=0$ for each $i\in A$ such that $v_{z+1}$ prefers $l$ to $i$. Thus, $v_{z+1}$ vetoes a clone of~$l$ in $x^{z,s+1}$, and (\ref{claim_condition3})-(\ref{claim_condition5}) are true for $(z+1,s)$.

\medskip
\setlength{\leftskip}{0 em} Suppose the claim is true for some pair $(Z,s)$ with $s<S-1$. We will prove that the claim is true for $(0,s+1)$. By (\ref{claim_condition1})-(\ref{claim_condition3}) for $(Z,s)$, $x^{Z,s+1}(1)\geq x^{Z,s}(1)>0$, so (\ref{claim_condition1}) is true for $(0,s+1)$. By~(\ref{xzs_inequality}), $x^{0,s+2}=x^{0,s+1}$ and (\ref{claim_condition2}) is true for $(0,s+1)$. 

\end{proof}

\begin{corollary}\label{cor:allmonotonic}%SUGGESTED BY ALEKSEI
Veto by consumption, fractional veto, and voting by veto tokens are monotonic.
\end{corollary}

\section{Properties of specific core-consistent rules}\label{app:specificrules}

\begin{restatable}{proposition}{IterativeCoreExpectedSize}
 For each profile with $n$ voters and $m$ candidates, the size of iterative core does not exceed the minimum of $n$ and~$m$.
 
 With the Impartial Culture assumption and a fixed number of candidates $m$, as $n\ria\infty$, the expected size of iterative core tends to~1.  
\end{restatable}
\begin{proof}
The first statement follows from the fact that if the number of candidates exceeds the number of voters, then each voter's least preferred candidate is blocked.

The second statement we prove as follows. Fix $m$ and a subset of candidates $A$, where $2\leq |A|\leq m$. The size of the veto core is $|A|$ (i.e., no candidate is blocked) in a profile restricted to $A$ only if each candidate has $n/|A|$ last positions, the probability of which tends to 0 as $n\ria\infty$ by the de Moivre-Laplace Theorem. Then the probability that the size of veto core is $|A|$ in at least one of $2^m-1-m$ such subsets also tends to~$0$. Hence, the size of iterative core is 1 with probability $P_n$, which tends to~1. The expected size of iterative core does not exceed $P_n+m(1-P_n)$, which tends to~1.
\end{proof}

\begin{proposition}
    $\CR_{\textup{Lex}}$ does not satisfy negative involvement. 
\end{proposition}
\begin{proof}
    Consider the following voters:
\begin{itemize}
    \item $d\succ_1 b\succ_1 c\succ_1 a$
    \item $c\succ_2 d\succ_2 b\succ_2 a$
    \item $d\succ_3 a\succ_3 c\succ_3 b$
    \item $b\succ_4 a\succ_4 c\succ_4 d$
    \item $a\succ_5 d\succ_5 b\succ_5 c$
\end{itemize}
In the profile consisting of voters 1--4 the veto core is $\set{b,c,d}$, and $b$ is the $\CR_\text{Lex}$ winner. If we add voter 5, the veto core becomes $\set{c,d}$ and the $\CR_\text{Lex}$ winner is $c$, which violates negative involvement since voter 5 ranks $c$ last.
\end{proof}

\begin{proposition}
$\CR_{\textup{GA}}$ does not satisfy negative involvement. 
\end{proposition}
\begin{proof}
Consider the following profile:
\begin{itemize}
    \item 4 voters $c\succ b\succ a\succ d$.
    \item 2 voters $a\succ c\succ b\succ d$.
    \item 2 voters $b\succ d\succ a\succ c$.
    \item 1 voter $d\succ c\succ a\succ b$.
    \item 1 voter $d\succ a\succ b\succ c$.
    \item 1 voter $d\succ c\succ b\succ a$.
    \item 1 voter $c\succ a\succ b\succ d$.
\end{itemize}
It takes four voters for a veto power of one, seven for a veto power of two, ten for a veto power of three. Candidate $d$ is blocked, the core is $\set{a,b,c}$.

Candidate $c$ has more last places, $a$ and $b$ have the same number of last places, but $b$ has less second-to-last places, so $b$ wins.

Add a voter of type $cbda$. It still takes ten voters to have a veto power of three, so the 4 voters of type $cbad$, 2 voters $acbd$, 1 voter $dcab$, 1 voter $dcba$, 1 voter $cabd$, and the new voter $cbda$ block $b$ in favour of~$c$. The core is $\set{a,c}$, but $a$ has less last places than $c$ and wins.
\end{proof}

\begin{definition}
\emph{Up-front tokens} (UFT) is voting by veto tokens where $r,t$ return the smallest integers satisfying the provisos of Lemma~\ref{lem:moulin}, and $L(r(m,n),n)$ be the order that places all $r(m,n)$ tokens of voter 1 first, followed by all $r(m,n)$ tokens of voter 2, and so on (see \autoref{ex:tokens}).
\end{definition}

\begin{proposition}
Up-front tokens is not positive responsive.
\end{proposition}
\begin{proof}
Consider the profile from \autoref{ex:tokens}:
\begin{itemize}
\item $c\succ_1 a\succ_1 b$.
\item $a\succ_2 b\succ_2 c$.
\item $b\succ_3 a\succ_3 c$.
\end{itemize}
At the end of the procedure $a$ has 2 clones remaining and $b$ has 1, so the outcome is a tie between $a$ and $b$. If voter 1 were to change his type to $acb$ he would still use all his veto tokens on $b$, and the outcome would be the same.
\end{proof}

\begin{proposition}
Up-front tokens does not satisfy independence of unanimous losers.
\end{proposition}
\begin{proof}
Consider the profile consisting of voter one of type $ab$ and voter two of type $ba$. The procedure will make 5 copies of every candidate and give 4 veto tokens to each voter, so the outcome is the tie $\set{a,b}$. Now add the unanimous loser $c$ to obtain voter one of type $abc$ and voter two of type $bac$. Now $t(3,2) = 3$ and $r(3,2)=4$. Voter one will eliminate all copies of $c$ and one copy of $b$, while voter two will eliminate all copies of $a$ and hence $b$ will win.
\end{proof}

\begin{proposition}
Up-front tokens does not satisfy positive involvement.
\end{proposition}
\begin{proof}
Consider the voters:
\begin{itemize}
\item $a\succ_1 b\succ_1 c$.
\item $b\succ_2 a\succ_2 c$.
\item $c\succ_3 b\succ_3 a$.
\item $a\succ_4 b\succ_4 c$.
\end{itemize}
Suppose we have a profile with voters 2--4. In this case $r(3,3) = 9$ and $t(3,3)=10$, so at the end of voting $a$ has one clone and $b$ has two clones remaining, the outcome is $\set{a,b}$.

Now consider the profile with all four voters. Now $r(3,4)=5$, $t(3,4) = 7$. After voting $b$ has one clone remaining, and the outcome is $\set{b}$, which violates positive involvement since voter 1 ranks $a$ first. 
\end{proof}

\begin{proposition}
Up-front tokens does not satisfy voter adaptability.
\end{proposition}
\begin{proof}
Consider the voters:
\begin{itemize}
\item $a\succ_1 c\succ_1 d \succ_1 b$.
\item $d\succ_2 c\succ_2 a \succ_2 b$.
\item $a\succ_3 b\succ_3 c \succ_3 d$.
\item $a\succ_4 b\succ_4 c \succ_3 d$.
\end{itemize}
In the profile consisting of voters 2--4, $r(4,3)=5, t(4,3)=4$ and $a$ is the unique winner. If we add voter 1, $r(4,4)=16$, $t(4,4)=17$. Voter one vetoes 16 clones of $b$, voter two 1 clone of $b$ and 15 clones of $a$, voter three 16 clones of $d$, voter four 1 clone of $d$ and 15 clones of $c$. In the end $a$ and $c$ have two clones each, and the outcome is $\set{a,c}$.
\end{proof}

\bibliographystyle{apalike}
\bibliography{arxiv}

\begin{thebibliography}{}

\bibitem[Anbarci, 2006]{Anbarci06}
Anbarci, N. (2006).
\newblock Finite alternating-move arbitration schemes and the equal area
  solution.
\newblock {\em Theory and Decision}, 61(1):21--50.

\bibitem[Barber{\`a} and Coelho, 2022]{BarberaCoelho22}
Barber{\`a}, S. and Coelho, D. (2022).
\newblock Compromising on compromise rules.
\newblock {\em The RAND Journal of Economics}, 53(1):95--112.

\bibitem[Black, 1958]{Black58book}
Black, D. (1958).
\newblock {\em The Theory of Committees and Elections}.
\newblock Cambridge University Press.
\newblock Reprinted in 1987 by Kluwer Academic Publishers.

\bibitem[Bogomolnaia and Moulin, 2001]{BogomolnaiaMoulin01}
Bogomolnaia, A. and Moulin, H. (2001).
\newblock A new solution to the random assignment problem.
\newblock {\em Journal of Economic Theory}, 100(2):295 -- 328.

\bibitem[Bol et~al., 2022]{BolLaslierNunez22}
Bol, D., Laslier, J.-F., and N{\'u}{\~n}ez, M. (2022).
\newblock {Two person bargaining mechanisms: A laboratory experiment}.
\newblock {\em Group Decision and Negotiation}, 31(6):1145--1177.

\bibitem[Bouveret et~al., 2017]{Bouveret2017}
Bouveret, S., Chevaleyre, Y., Durand, F., and Lang, J. (2017).
\newblock Voting by sequential elimination with few voters.
\newblock In {\em Proceedings of the 26th International Joint Conference on
  Artificial Intelligence, {IJCAI'2017}}, pages 128--134. AAAI Press.
\newblock Available at: \url{https://www.ijcai.org/proceedings/2017/0019.pdf}.

\bibitem[Brandt et~al., 2017]{BrandtGeistPeters17}
Brandt, F., Geist, C., and Peters, D. (2017).
\newblock {Optimal bounds for the no-show paradox via SAT solving}.
\newblock {\em Mathematical Social Sciences}, 90:18--27.

\bibitem[Bure and Parilina, 2013]{BureParilina13bookEnglish}
Bure, V.~M. and Parilina, E.~M. (2013).
\newblock {\em Teoriya Veroyatnostey i Matematicheskaya Statistika}.
\newblock Lan', Saint-Petersburg.
\newblock (In Russian).

\bibitem[Duggan and Schwartz, 1992]{Duggan1992}
Duggan, J. and Schwartz, T. (1992).
\newblock {Strategic manipulability is inescapable: Gibbard-Satterthwaite
  without resoluteness}.
\newblock Working Papers 817, California Institute of Technology, Division of
  the Humanities and Social Sciences.
\newblock Available at:
  \url{https://authors.library.caltech.edu/80867/1/sswp817.pdf}.

\bibitem[Duggan and Schwartz, 2000]{Duggan2000}
Duggan, J. and Schwartz, T. (2000).
\newblock {Strategic manipulability without resoluteness or shared beliefs:
  Gibbard-Satterthwaite generalized}.
\newblock {\em Social Choice and Welfare}, 17(1):85--93.

\bibitem[Felsenthal and Nurmi, 2018]{FelsenthalNurmi18book}
Felsenthal, D.~S. and Nurmi, H. (2018).
\newblock {\em Voting Procedures for Electing a Single Candidate: Proving Their
  (In) Vulnerability to Various Voting Paradoxes}.
\newblock Springer.

\bibitem[Fishburn, 1971]{Fishburn71}
Fishburn, P.~C. (1971).
\newblock A comparative analysis of group decision methods.
\newblock {\em Behavioral Science}, 16(6):538--544.

\bibitem[Fishburn, 1973]{Fishburn73book}
Fishburn, P.~C. (1973).
\newblock {\em The Theory of Social Choice}.
\newblock Princeton University Press.

\bibitem[Fishburn, 1982]{Fishburn82}
Fishburn, P.~C. (1982).
\newblock Monotonicity paradoxes in the theory of elections.
\newblock {\em Discrete Applied Mathematics}, 4(2):119--134.

\bibitem[Fishburn and Brams, 1983]{FishburnBrams83}
Fishburn, P.~C. and Brams, S.~J. (1983).
\newblock Paradoxes of preferential voting.
\newblock {\em Mathematics Magazine}, 56(4):207--214.

\bibitem[Garey and Johnson, 1979]{Garey1979book}
Garey, M.~R. and Johnson, D.~S. (1979).
\newblock {\em Computers and Intractability: A Guide to the Theory of
  NP-Completeness}.
\newblock W. H. Freeman \& Co., New York, NY, USA.

\bibitem[Goldsmith et~al., 2014]{Goldsmith14}
Goldsmith, J., Lang, J., Mattei, N., and Perny, P. (2014).
\newblock Voting with rank dependent scoring rules.
\newblock In {\em Proceedings of the AAAI Conference on Artificial
  Intelligence}, volume~28.
\newblock Available at: \url{https://doi.org/10.1609/aaai.v28i1.8826}.

\bibitem[Holliday and Pacuit, 2023]{HollidayPacuit23pc}
Holliday, W.~H. and Pacuit, E. (2023).
\newblock {Split Cycle: A new Condorcet consistent voting method independent of
  clones and immune to spoilers}.
\newblock {\em Public Choice}.
\newblock Forthcoming. Available at: \url{https://arxiv.org/abs/2004.02350}.

\bibitem[Ianovski and Kondratev, 2021]{IanovskiKondratev21}
Ianovski, E. and Kondratev, A.~Y. (2021).
\newblock Computing the proportional veto core.
\newblock In {\em Proceedings of the AAAI Conference on Artificial
  Intelligence}, volume~35, pages 5489--5496.
\newblock Available at: \url{https://doi.org/10.1609/aaai.v35i6.16691}.

\bibitem[Jimeno et~al., 2009]{JimenoPerezGarcia09}
Jimeno, J.~L., P{\'e}rez, J., and Garc{\'\i}a, E. (2009).
\newblock {An extension of the Moulin No Show Paradox for voting
  correspondences}.
\newblock {\em Social Choice and Welfare}, 33(3):343--359.

\bibitem[Kondratev et~al., 2019]{KondratevIanovskiNesterov19}
Kondratev, A.~Y., Ianovski, E., and Nesterov, A.~S. (2019).
\newblock How should we score athletes and candidates: geometric scoring rules.
\newblock {\em arXiv preprint}.
\newblock Available at: \url{https://arxiv.org/abs/1907.05082}.

\bibitem[Kondratev and Nesterov, 2020]{KondratevNesterov2020}
Kondratev, A.~Y. and Nesterov, A.~S. (2020).
\newblock Measuring majority power and veto power of voting rules.
\newblock {\em Public Choice}, 183:187--210.

\bibitem[Laslier et~al., 2021]{LaslierNunezSanver21}
Laslier, J.-F., N{\'u}{\~n}ez, M., and Sanver, M.~R. (2021).
\newblock A solution to the two-person implementation problem.
\newblock {\em Journal of Economic Theory}, 194:105261.

\bibitem[Luce and Raiffa, 1957]{LuceRaiffa57book}
Luce, R.~D. and Raiffa, H. (1957).
\newblock {\em Games and Decisions: Introduction and Critical Survey}.
\newblock John Wiley and Sons, New York.

\bibitem[Malhotra et~al., 1978]{Malhotra1978}
Malhotra, V., Kumar, M., and Maheshwari, S. (1978).
\newblock An {$O(|V|^3)$} algorithm for finding maximum flows in networks.
\newblock {\em Information Processing Letters}, 7(6):277--278.

\bibitem[Meredith, 1913]{Meredith1913book}
Meredith, J.~C. (1913).
\newblock {\em Proportional Representation in Ireland}.
\newblock Edward Ponsonby, Ltd, Dublin.

\bibitem[Moulin, 1981a]{Moulin1981}
Moulin, H. (1981a).
\newblock The proportional veto principle.
\newblock {\em The Review of Economic Studies}, 48(3):407--416.

\bibitem[Moulin, 1981b]{Moulin1981a}
Moulin, H. (1981b).
\newblock Prudence versus sophistication in voting strategy.
\newblock {\em Journal of Economic Theory}, 24(3):398--412.

\bibitem[Moulin, 1983]{Moulin1983book}
Moulin, H. (1983).
\newblock {\em The Strategy of Social Choice}.
\newblock Advanced textbooks in economics. North-Holland Publishing Company.

\bibitem[Moulin, 1988]{Moulin1988}
Moulin, H. (1988).
\newblock Condorcet's principle implies the no show paradox.
\newblock {\em Journal of Economic Theory}, 45(1):53--64.

\bibitem[Mueller, 1978]{Mueller1978}
Mueller, D.~C. (1978).
\newblock Voting by veto.
\newblock {\em Journal of Public Economics}, 10(1):57--75.

\bibitem[N{\'u}{\~n}ez and Sanver, 2017]{NunezSanver17}
N{\'u}{\~n}ez, M. and Sanver, M.~R. (2017).
\newblock Revisiting the connection between the no-show paradox and
  monotonicity.
\newblock {\em Mathematical Social Sciences}, 90:9--17.

\bibitem[P{\'e}rez, 2001]{Perez01}
P{\'e}rez, J. (2001).
\newblock {The strong no show paradoxes are a common flaw in Condorcet voting
  correspondences}.
\newblock {\em Social Choice and Welfare}, 18(3):601--616.

\bibitem[Richelson, 1978]{Richelson78III}
Richelson, J.~T. (1978).
\newblock {A comparative analysis of social choice functions, III}.
\newblock {\em Behavioral Science}, 23(3):169--176.

\bibitem[Saari, 1989]{Saari89jet}
Saari, D.~G. (1989).
\newblock A dictionary for voting paradoxes.
\newblock {\em Journal of Economic Theory}, 48(2):443--475.

\bibitem[Saari, 1994]{Saari94book}
Saari, D.~G. (1994).
\newblock {\em Geometry of Voting}.
\newblock Springer.

\bibitem[Schulze, 2011]{Schulze11}
Schulze, M. (2011).
\newblock A new monotonic, clone-independent, reversal symmetric, and
  {C}ondorcet-consistent single-winner election method.
\newblock {\em Social Choice and Welfare}, 36(2):267--303.

\bibitem[Smith, 1973]{Smith73}
Smith, J.~H. (1973).
\newblock Aggregation of preferences with variable electorate.
\newblock {\em Econometrica}, 41(6):1027--1041.

\bibitem[Taylor, 2002]{Taylor2002}
Taylor, A.~D. (2002).
\newblock The manipulability of voting systems.
\newblock {\em The American Mathematical Monthly}, 109(4):321--337.

\bibitem[Tideman, 1987]{Tideman1987}
Tideman, T.~N. (1987).
\newblock Independence of clones as a criterion for voting rules.
\newblock {\em Social Choice and Welfare}, 4(3):185--206.

\bibitem[Young, 1975]{Young75}
Young, H.~P. (1975).
\newblock Social choice scoring functions.
\newblock {\em SIAM Journal on Applied Mathematics}, 28(4):824--838.

\bibitem[Young and Levenglick, 1978]{YoungLevenglick78}
Young, H.~P. and Levenglick, A. (1978).
\newblock {A consistent extension of Condorcet’s election principle}.
\newblock {\em SIAM Journal on Applied Mathematics}, 35(2):285--300.

\bibitem[Zwicker, 2016]{Zwicker16}
Zwicker, W.~S. (2016).
\newblock Introduction to the theory of voting.
\newblock In Brandt, F., Conitzer, V., Endriss, U., Lang, J., and Procaccia,
  A.~D., editors, {\em Handbook of Computational Social Choice}, pages 23--56.
  Cambridge University Press, New York.

\end{thebibliography}
\end{document}